\documentclass{llncs}
\usepackage[utf8]{inputenc}
\usepackage{amsmath,amssymb,amsfonts}
\usepackage{listings, color, graphicx, float}
\usepackage[hidelinks]{hyperref}
\usepackage[nameinlink,capitalize]{cleveref}
\usepackage{textcomp}
\usepackage{multirow}
\usepackage{enumitem}
\usepackage{comment}
\usepackage{mdframed}
\usepackage{xcolor}
\usepackage{mathrsfs}
\usepackage{breakcites}

\hypersetup{
    colorlinks=true,
    citecolor=blue,
    linkcolor=blue,
    filecolor=cyan,
    urlcolor=magenta,
}

\crefname{fact}{Fact}{Facts}

\newcommand{\N}{\mathbb{N}}
\newcommand{\C}{\mathbb{C}}
\newcommand{\F}{\mathbb{F}}

\newcommand{\R}{\mathbb{R}}

\DeclareMathOperator{\Tr}{Tr}
\newcommand{\LG}{{\sf LG}}
\newcommand{\DG}{{\mathsf{DG}}}
\newcommand{\WG}{{\mathsf{WG}}}
\newcommand{\AG}{{\mathsf{AG}}}

\newcommand{\ket}[1]{\left| #1 \right\rangle}
\newcommand{\bra}[1]{\left\langle #1 \right|}

\def\P{\textsf{P}}

\def\iO{\textsf{iO}}
\def\shO{\textsf{shO}}
\def\poly{\textsf{poly}}

\newcommand{\cS}{{\cal S}}
\newcommand{\cD}{{\mathcal{D}}}
\newcommand{\cN}{{\mathcal{N}}}
\newcommand{\cE}{{\mathcal{E}}}
\newcommand{\cR}{{\mathcal{R}}}
\newcommand{\shift}{{\sf Shift}}
\newcommand{\api}{{\sf API}}
\newcommand{\ati}{{\sf ATI}}

\newcommand{\ti}{{\sf TI}}
\newcommand{\projimp}{{\sf ProjImp}}
\newcommand{\cproj}{{\sf CProj}}

\newcommand{\cP}{{\mathcal{P}}}
\newcommand{\cM}{{\cal M}}
\newcommand{\cA}{{\cal A}}

\newcommand{\cO}{{\cal O}}
\newcommand{\cF}{{\cal F}}
\newcommand{\cC}{\mathcal{C}}
\newcommand{\cH}{\mathcal{H}}

\newcommand{\negl}{{\mathrm{\textsf{negl}}}}
\newcommand{\neglfunc}{{\mathrm{\textsf{negl}}}}
\newcommand{\qbin}[2]{\begin{bmatrix}#1\\#2\end{bmatrix}}

\newcommand{\eval}{\mathsf{Eval}}

\newcommand{\Dist}{\textsf{Dist}}

\newcommand{\keygen}{\mathsf{KeyGen}}

\newcommand{\sign}{\textsf{Sign}}

\newcommand{\ver}{\textsf{Ver}}

\newcommand{\sk}{\mathsf{sk}}

\newcommand{\pk}{\mathsf{pk}}

\newcommand{\gennote}{\textsf{GenNote}}

\newcommand{\compute}{\textsf{Compute}}

\newcommand{\generate}{\textsf{Generate}}

\newcommand{\enc}{\mathsf{Enc}}

\newcommand{\dec}{\mathsf{Dec}}

\newcommand{\ct}{\textsf{ct}}

\newcommand{\pke}{\mathsf{PKE}}

\newcommand{\cd}{\mathsf{CD}} 

\newcommand{\eventone}{\mathsf{E_{1}}}
\newcommand{\bareventone}{\mathsf{\bar{E}_1}}
\newcommand{\eventtwo}{\mathsf{E_{2}}}
\newcommand{\bareventtwo}{\mathsf{\bar{E}_2}}

\newcommand{\sig}{{\sf sig}}
\newcommand{\As}{\mathcal{A}}
\newcommand{\Bs}{\mathcal{B}}
\newcommand{\Cs}{\mathcal{C}}

\newcommand{\mar}{{\sf Mark}}
\newcommand{\extract}{{\sf Extract}}

\newcommand{\prf}{{\sf PRF}}
\newcommand{\compile}{{\sf Compile}}

\newcommand{\gen}{{\sf Gen}}
\newcommand{\chk}{{\sf Check}}

\newcommand{\setup}{{\sf Setup}}
\newcommand{\samp}{{\sf Samp}}
\newcommand{\qm}{{\sf QM}}
\newcommand{\wm}{{\sf WM}}

\newcommand{\vk}{{\sf vk}}
\newcommand{\xk}{{\sf xk}}
\newcommand{\mk}{{\sf mk}}

\newcommand{\wpp}{{\sf wpp}}

\newcommand{\aux}{{\sf aux}}
\newcommand{\AUX}{{\sf AUX}}

\definecolor{mygreen}{RGB}{80,180,0}
\definecolor{b2}{RGB}{51,153,255}

\newcommand{\Ruizhe}[1]{{\color{red}[Ruizhe: #1]}}

\newcommand{\Jiahui}[1]{{\color{b2}[Jiahui: #1]}}

\newcommand{\Qipeng}[1]{{\color{magenta}[Qipeng: #1]}}

\mdfdefinestyle{figstyle}{ 
  linecolor=black!7, 
  backgroundcolor=black!7, 
  innertopmargin=10pt, 
  innerleftmargin=25pt, 
  innerrightmargin=25pt, 
  innerbottommargin=10pt 
}

\newenvironment{gamespec}{
  \begin{mdframed}[style=figstyle]}{
  \end{mdframed}}

\makeatletter
\renewcommand\paragraph{\@startsection{paragraph}{4}{\z@}%
                     {-12\p@ \@plus -4\p@ \@minus -4\p@}%
                     {-0.5em \@plus -0.22em \@minus -0.1em}%
                     {\normalfont\normalsize\bfseries}}
\makeatother

\usepackage{authblk}
\date{}

\begin{document}
\pagestyle{plain}

\title{New Approaches for Quantum Copy-Protection}

\author{}
\institute{}
\author{Scott Aaronson\inst{1} \and
Jiahui Liu\inst{1} \and
Qipeng Liu\inst{3} \and
Mark Zhandry\inst{2} \and 
Ruizhe Zhang\inst{1}}

\institute{The University of Texas at Austin \\
\email{\{aaronson, jiahui, rzzhang\}@cs.utexas.edu} \and
Princeton University \& NTT Research, USA \\
\email{mzhandry@princeton.edu} \and
Princeton University, USA \\
\email{qipengl@princeton.edu}
}
\maketitle

\begin{abstract}
   Quantum copy protection uses the unclonability of quantum states to construct quantum software that provably cannot be pirated. Copy protection would be immensely useful, but unfortunately little is known about how to achieve it in general. In this work, we make progress on this goal, by giving the following results:
\begin{itemize}
    \item We show how to copy protect any program that cannot be learned from its input/output behavior, relative to a \emph{classical} oracle. This improves on Aaronson [CCC'09], which achieves the same relative to a quantum oracle. By instantiating the oracle with post-quantum candidate obfuscation schemes, we obtain a heuristic construction of copy protection.
    \item We show, roughly, that any program which can be watermarked can be copy \emph{detected}, a weaker version of copy protection that does not prevent copying, but guarantees that any copying can be detected. Our scheme relies on the security of the assumed watermarking, plus the assumed existence of public key quantum money. Our construction is general, applicable to many recent watermarking schemes.
\end{itemize}

\end{abstract}

\section{Introduction}
\paragraph{ }

Quantum copy-protection, proposed by Aaronson~\cite{aaronson2009quantum}, aims to use the unclonability of quantum states to achieve programs that cannot be copied. That is, the program $f$ is given as a quantum state $\ket{\psi_f}$. $\ket{\psi_f}$ allows for computing $f$ on arbitrary inputs; meanwhile, it is infeasible to copy the state $\ket{\psi_f}$, or even convert $\ket{\psi_f}$ into two arbitrary states that both allow for computing $f$. The quantum no-cloning theorem shows that quantum states in general cannot be copied. Copy protection takes this much further, augmenting the unclonable state with the ability to evaluate programs. Copy-protection would have numerous applications to intellectual property management, and to cryptography generally.

Progress on quantum copy-protection has unfortunately been slow. On the negative side, copy-protection for general programs is impossible. As explained by Aaronson, any \emph{learnable} program---that is, a program whose description can be learned from just it's input/output behavior---cannot be copy-protected. Indeed, an attacker, given the (copy-protected) code for the program can just query the code on several inputs, and learn the original program from the results. The original program can then be copied indefinitely. A more recent result of Ananth and La Placa~\cite{ananth2020secure} shows, under certain computational assumptions, that even certain contrived unlearnable programs cannot be copy-protected.

On the positive side, Aaronson demonstrates a quantum oracle\footnote{That is, an oracle that actually implements a quantum operation.} relative to which copy-protection exists for any unlearnable program. Due to the negative result above, this scheme cannot be instantiated in the general. Worse, even for programs that are not subject to the impossibility result, it remains unclear how to even heuristically instantiate the scheme. Very recently, Ananth and La Placa~\cite{ananth2020secure} build a version of copy protection which they call software leasing, which guarantees a sort of copy \emph{detection} mechanism: unfortunately, their work explicitly allows copying the functionality and only guarantees that such copying can be detected. Also, their construction only works for a certain class of ``evasive'' functions, which only accept a hidden sparse set of inputs. The work of Ben-David and Sattath~\cite{ben2016quantum} and more recently Amos et al.~\cite{AGKZ20} can be seen as copy-protecting very specific cryptographic functionalities.

\subsection{This Work}

In this work, we give new general results for copy protection. Our two main results are:
\begin{itemize}
    \item Any unlearnable functionality can be copy-protected, relative to a \emph{classical} oracle.
    \item Any functionality that can be \emph{watermarked} in a certain sense, can be copy-\emph{detected} assuming just the existence of public key quantum money.
\end{itemize}
Both of our results are very general, applying to a wide variety of learning and watermarking settings, including settings where functionality preservation is not required. Along the way to obtaining our results, we give new definitions for security of copy-protection (as well as copy detection and watermarking), which provide for much stronger guarantees.

Our first result improves Aaronson~\cite{aaronson2009quantum} to use a classical oracle, which can then heuristically be instantiated using candidate post-quantum obfuscation (e.g.~\cite{BGMZ18,brakerski-garg20iofromlwe}), resulting in a concrete candidate copy-protection scheme. Of course, the impossibility of Ananth and La Placa~\cite{ananth2020secure} means the resulting scheme cannot be secure in the standard model for arbitrary programs, but it can be conjectured to be secure for programs not subject to the impossibility.

Our second result complements Ananth and La Placa~\cite{ananth2020secure}'s positive result for copy-detecting evasive functions, by copy-detecting arbitrary watermarkable functions. For our purposes, watermarkable functions are those that can have a publicly observable ``mark'' embedded into the program, such that it is infeasible to remove the mark without destroying the functionality. We note that the results (and techniques) are incomparable to~\cite{ananth2020secure}. First, watermarkable functions are never evasive, so the class of functions considered are disjoint. Second, our security guarantee is much stronger than theirs, which we discuss in Section~\ref{sec:overview}. 

Taken together, we believe our results strongly suggest that watermarkable functions may be copy-protectable. Concretely, the impossibility result of Ananth and La Placa also applies to copy detection, and our second result shows that watermarkable functions therefore circumvent the impossibility. Based on this, we conjecture that our first result, when instantiated with candidate obfuscators, is a secure copy-protection scheme for watermarkable functions. We leave proving or disproving our conjecture as an interesting direction for future work.

\subsection{Technical Overview}\label{sec:overview}

\paragraph{Definitional Work.} We first look at one attempt of defining quantum copy-protection. We say an adversary successfully pirates a quantum program for computing function $f$, if it outputs two quantum programs $\sigma_1, \sigma_2$, each of them able to compute $f$ correctly with probability greater than some threshold.  
Consider the following case. Let $f$ be a signing algorithm with a particular signing key hard-coded. Suppose that there are many valid signatures for each message. Consider a hypothetical adversary which “splits” the program into two pieces, each computing valid signatures, but neither computing the same signature that $f$ produces. Such programs are ``good enough'' for many applications, but this adversary would not be ruled out by the usual security notions.

Another example is copy-protection of public key encryption. Let $f$ be a decrypting algorithm with a particular decryption key hard-coded. Suppose the split two program pieces only work correctly on a sparse set: namely they can only decrypt correctly on ciphertexts of $m_0, m_1$; for ciphertexts of other messages, they output junk. This splitting attack does not violate the security notion either, since both functions produced by the adversary differ from the original program on most inputs.  But again, such programs are ``good enough'' for some applications. 

Similar definitonal issues were discussed in~\cite{goyal2019watermarking}, but in the context of watermarking primitives. As we will see, watermarking is closely related to copy-detection and copy-protection.

Our solution is to define ``compute $f$ correctly'' by a relation. The relation takes some random coins $r$, the function $f$ (with some additional information about $f$ hard-coded in the circuit); it samples an input and runs the (quantum) program on that classical input; finally, it checks the output of the quantum program, testing in superposition if the output $z$ together with $f, r$ is in the relation. As an example, if $f$ is a signing circuit (with the verification key hard-coded), the relation is defined as: use random coins $r$ to generate a random message $m$, run the quantum program on $m$ and test in superposition if it is a valid signature, by applying the verification algorithm $\ver(\vk, m, \cdot\,; r)$. 

Unfortunately, formalizing these other definitions can still be tricky. For example, we want that the adversary can not take a program for $f$ and produce two programs that each computes $f$ correctly on half inputs of the domain. In this setting, we would naturally say that a program is good if it correctly computes the function $f$ with probability $1/2$. However the definition becomes problematic. Consider the adversary which takes its quantum program $P$ and simply produces $\frac{1}{\sqrt{2}}(\ket P \ket D \ket 0 + \ket D \ket P \ket 1)$ where $D$ is a dummy program that outputs junk. Now, the two halves of this bipartite system each has probability $1/2$ of outputting the right answer on a random input. Thus, both halves would naturally be considered to compute correctly, according to this definition. Therefore, any security definition like this is trivially false.

For another example, consider the adversary produces  $\frac{1}{\sqrt{3}}\ket P \ket P + \frac{\sqrt{2}}{\sqrt{3}} \ket D \ket D$. The two halves of this bipartite system each has probability $1/3$ of outputting the right answer on a random input. However, both halves can successfully answer all inputs correctly at the same time, with probability $1/3$. Thus, it is secure under the security definition above, but the adversary actually perfectly pirates the program with some constant probability.

Our solution will be to use recent ideas from Zhandry~\cite{z20}, who considered similar issues in the context of traitor tracing. At a high level, the issue above is that we are trying to assign a property to a quantum state (whether the state is a good program), but this property is non-physical and does not make sense for mixed or entangled states. Instead, we want ``a program is good'' to be a measurement that can be applied to the state. We would naturally also want the measurement to be projective, so that if a program  is once tested to be ``good'', it will always be ``good''.

Let $\cM = (M_0, M_1)$ be binary positive operator valued measures (POVMs) that represents choosing random coins and testing if the quantum program computes correctly with respect to the random coins. For a mixed quantum program state $\sigma$, the probability the program evaluates correctly relative to this test is given as ${\sf Tr}[M_0 \sigma]$ . Let $\cM'$ be the (inefficient) projective measurement $\{P_p\}_{p \in [0,1]}$, projecting onto the eigenspaces of $M_0$, where $p$ ranges over the corresponding eigenvalues of $M_0$\footnote{Since $M_0+M_1$ is the identity, $M_1$ shares the same eigenvectors, with eigenvalue $1-p$.} 
Zhandry showed that the measurement below results in an equivalent POVM as $\cM$:
\begin{itemize}
    \item Apply the projective measurement $\cM'$, and obtain $p$;
    \item Output $0$ with probability $p$, and output $1$ with probability $1-p$. 
\end{itemize}
Intuitively, $\cM'$ will project a state to a eigenvector with eigenvalue $p$, the state computes correctly on  $p$-fraction of all inputs. 

Therefore, we say a quantum program $\sigma$ is tested to be $\gamma$-good, if the measurement $\cM'$ has outcome $p \geq \gamma$. We say an adversary successfully pirates a quantum program for computing $f$, if the two programs are both tested to be $\gamma$-good with non-negligible probability. Using similar ideas, we define quantum unlearnability of programs, and quantum copy-detection.

\paragraph{Our Copy-Protection Scheme.}
We give a quantum copy-protection construction  for all unlearnable functions based on (1) classical oracles, and (2) subspace membership oracles, or more abstractly, any tokenized signature scheme~\cite{ben2016quantum}. 

A tokenized signature generates a signature token $\ket {\sf sig}$ which we call a signing token. A signer who gets one copy of the signing token can sign a single bit $b$ of her choice.  $\sign(b, \ket {\sf sig})$ outputs a classical signature whose correctness guarantee is the same as classical signatures: namely, verification will accept the result as a signature on $b$.  Importantly, the signing procedure is a unitary and will produce a superposition of all valid signatures of $b$; to obtain a classical signature, a measurement to the state is necessary which leads to a collapse of the token state.
Thus, a signature token $\ket {\sf sig}$ can only be used to produce one classical signature of a single bit and any attempt to produce a classical signature of the other bit would fail. 
\cite{ben2016quantum} formalizes this idea and constructs a tokenized signature scheme relative to a classical oracle (a subspace membership oracle).

The high-level idea of our copy-protection scheme is that it requires any authorized user to query an oracle twice on signatures of bits 0 and 1. Let $f$ be the function we want to copy-protect. Define the following circuits: 
\begin{align*}
    \cO_1(x, \sig) &= \begin{cases}
                            H(x) & \text{ if }  \ver(\vk, 0, \sig) = 1 \\
                            \bot & \text{ otherwise }
                    \end{cases} \\
    \cO_2(x, \sig) &= \begin{cases}
                            f(x) \oplus H(x) & \text{ if }  \ver(\vk, 1, \sig) = 1 \\
                            \bot & \text{ otherwise }
                    \end{cases}
\end{align*}
Here $H$ is a random function. The copy-protected program of $f$ is a signature token $\ket \sig$ and obfuscations of $\cO_1, \cO_2$, which we will heuristically treat as oracles to $\cO_1, \cO_2$. We denote this program as $(\ket \sig, \cO_1, \cO_2)$. 

To obtain $f(x)$, a user has to query on  signatures of both bits and get $H(x)$ and $H(x) \oplus f(x)$.  Note that even if with token $\ket \sig$ one can only produce one of the classical signatures, a user can still query both oracles $\cO_1, \cO_2$ multiple times. To obtain $H(x)$, a user can simply compute the superposition of all valid signatures of $0$ by applying a unitary, and feed the quantum state together with $x$ to $\cO_1$. It then measures the output register. The user never actually measures the signature. Because the output register contains a unique output $H(x)$, by Gentle Measurement Lemma \cite{aaronson2004limitations}, it can rewind the quantum state back to $\ket \sig$. Thus, our copy-protection scheme allows a copy-protected program to be evaluated on multiple inputs, multiple times. 

We next show how to prove anti-piracy security. 
Let $\sigma_1, \sigma_2$ be two (potentially entangled) program states pirated by an adversary, which makes oracle access to both $\cO_1, \cO_2$ and breaks the anti-piracy security. Let $\cO_\bot$ be an oracle that always outputs $\bot$.  If $\sigma_1$ never queries the oracle $\cO_2$, we know the two programs $(\sigma_1, \cO_1, \cO_2)$ and $(\sigma_1, \cO_1, \cO_\bot)$ would have almost identical output distribution. Moreover, $(\sigma_1, \cO_1, \cO_\bot)$ can be simulated even without querying $f$ because $\cO_1$ is simply a random oracle (on valid inputs).  Therefore, the program can be used to break the unlearnability of $f$. 
Similarly, if $\sigma_2$ never queries the oracle $\cO_1$, the program $(\sigma_2, \cO_\bot, \cO_2)$ can be used to break the unlearnability of $f$. 

Since $f$ is unlearnable, the above two cases can not happen.  We show under this case, we can extract signatures of 0 and 1. Intuitively, since $(\sigma_1, \cO_1, \cO_2)$ makes queries to $\cO_2$, we can run the program on random inputs and measure a random query to $\cO_2$, thereby extracting a signature of 1. Similarly it holds for $(\sigma_2, \cO_1, \cO_2)$ and one could extract a signature of 0. Unfortunately, this intuition does not quite work since $\sigma_1$ and $\sigma_2$ are potentially entangled. This means there can be correlations between the outcomes of the measurements producing the two signatures: perhaps, if the measurement on $(\sigma_1, \cO_1, \cO_2)$ produces a valid signature on 1, then the measurement on $(\sigma_2, \cO_1, \cO_2)$ is guaranteed to fail to produce a signature. We show by a delicate argument that in fact adversaries cannot cheat using such correlations.

\paragraph{Our Copy-Detection Scheme.}  We construct a copy-detection scheme for any function family that can be watermarked.  A watermarking scheme roughly consists the following procedure: $\mar$ takes a circuit and a message, and outputs a circuit embedded with that mark; $\extract$ takes a marked circuit and outputs the embedded mark. 
A watermarking scheme requires: (1) the watermarked circuit $\tilde{f} = \mar(f, m)$ should preserve its intended functionality as $f$; (2) any efficient adversary given a marked $\tilde{f}$, can not generate a new marked circuit $\hat{f}$ with a different mark, while preserving its functionality. 
Watermarking primitives have been studied in previous works including \cite{cohen2018watermarking,kim2017watermarking,quach2018watermarking,kim2019watermarking,goyal2019watermarking}. 

Our construction also requires a public key quantum money scheme. It consists  two procedures: $\gen$ and $\ver$.  $\gen$ takes a security parameter and outputs a quantum banknote $\ket \$$. $\ver$ is public, takes a quantum money banknote, and outputs either a serial number of that banknote or $\bot$ indicating it is an invalid banknote. The security requires no efficient adversary could use $\ket \$$ to prepare $\ket{\$_1} \ket {\$_2}$ such that both banknotes pass the verification and their serial numbers are equal to that of $\ket \$$. We note that this version of quantum money corresponds to a ``mini-scheme'' as defined by~\cite{aaronson2012quantum}.

The copy-detection scheme takes a function $f$, samples a banknote $\ket \$$ with serial number $s$, lets $\tilde{f} \gets \mar(f, s)$ and outputs the copy-detected program as $(\tilde{f}, \ket \$)$. To evaluate the function, it simply runs the classical program $\tilde{f}$. To check a program is valid, it  extracts the serial number from the money state and compares it with the mark of the program. 

The security requires that no efficient adversary could produce $\tilde{f}_1, \ket {\$_1}$ and $\tilde{f}_2, \ket {\$_2}$ such that two programs pass the check and both classical circuits preserve the functionality. Let $s$ be the serial number of $\ket {\$}$, $s_b$ be the serial number of $\ket {\$_b}$ for $b = 1, 2$. To pass the check, there are two possible cases: 
\begin{itemize}
    \item $s_1 = s_2 = s$. In this case, $ \ket {\$_1} \ket {\$_2}$ breaks the security of the quantum money scheme because one successfully duplicates a banknote with the same serial number. 
    \item At least one of $s_b \ne s$.  Because the mark of $\tilde{f}_b$ is also equal to $s_b$, one of $\tilde{f}_b$ breaks the security of the watermarking scheme, as it preserves the functionality, while having a different mark than $s$. 
\end{itemize}
We show the above construction and proof apply to a wide range of watermarking primitives. 

\paragraph{Copy-Protection in the Standard Model?} The security of our copy-protection scheme requires treating the obfuscated programs as oracles. While we prove security for all unlearnable programs, we cannot expect such security to hold in the standard model: as shown in~\cite{ananth2020secure}, there are unlearnable functions that can cannot be copy-protected, or even copy-detected. 
On the other hand, watermarkable programs are a natural class of programs that are necessarily immune to the style of counter-example of Barak et al.~\cite{barak2001possibility}, on which the copy-protection impossibility is based. Namely, the counter-example works by giving programs that are unlearnable, but such that having any (even approximate~\cite{BitPan15}) code for the program lets you recover the original program. Such programs \emph{cannot} be watermarkable, as the adversary can always recover the original program from the (supposedly) watermarked program. 

Thus, we broadly conjecture that all watermarkable functions can be copy-protected. Our copy-detection result gives some evidence that this may be feasible. Concretely, we conjecture that our copy-protection construction is secure for any watermarkable program, when the oracles are instantiated with post-quantum obfuscation constructions. We leave justifying either the broad or concrete conjectures as fascinating open questions.

\subsection{Other Related Works}
\paragraph{Quantum Copy Protection}
Quantum copy-protection was proposed by Aaronson in \cite{aaronson2009quantum}; this paper gave two candidate schemes for copy-protecting point functions without security proofs and showed that any functions that are not quantum learnable can be quantum copy-protected relative to a quantum oracle (an oracle which could perform an arbitrary unitary).

\cite{ananth2020secure} 
gave a conditional impossibility of general copy-protection: they construct a quantum unlearnable circuit using the quantum FHE scheme and compute-and-compare obfuscation~\cite{wichs2017obfuscating,goyal2017lockable} that is not copy-protectable once a QPT adversary has non-black-box access to the program. 
\cite{ananth2020secure} also gave a new definition that is weaker than the standard copy-protection security, called Secure Software Leasing (SSL) and an SSL construction for a subclass of evasive functions, namely, searchable compute-and-compare circuits.
 

  \cite{Broadbent2019UncloneableQE} introduced unclonable encryption. They construct schemes for encoding classical plaintexts into quantum ciphertexts, which prevents copying of encrypted data. Unclonable encryption can be seen as copy-protecting a unit of functional information simpler than a function.
 \cite{georgiou-zhandry20} introduced another new notion, unclonable decryption keys; in contrast to making the ciphertext unclonable as in \cite{Broadbent2019UncloneableQE}, they construct schemes where the decryption key is unclonable, therefore allowing only one decryptor to decrypt successfully at a time. 
A more recent work is \cite{coladangelo2020quantum}, giving a construction for copy-protecting point functions in the quantum random oracle model with techniques inspired by \cite{Broadbent2019UncloneableQE} and the construction can be extended to copy-protecting compute-and-compare circuits.

\paragraph{Quantum Money} 
Quantum money was first proposed by Wiesner in around 1970; \cite{wiesner1983conjugate} gave a first private-key quantum money scheme based on conjugate coding.  Aaronson \cite{aaronson2009quantum} gave a first public-key quantum money scheme; he proved that it is possible to construct the secure public-key quantum money relative to a quantum oracle. However, his explicit scheme was broken by Lutomirski et al. \cite{lutomirski2009breaking}. Later, Aaronson and Christiano \cite{aaronson2012quantum} proposed a secure public-key quantum money scheme relative to a classical oracle. Zhandry \cite{zhandry2017quantum} investigated a kind of collision-free quantum money called quantum lightning and the win-win relationship between the security of signatures/hash functions and quantum money; \cite{zhandry2017quantum} also instantiated the quantum money scheme of \cite{aaronson2012quantum} with quantum-secure indistinguishability obfuscation.  Kane \cite{kane2018quantum} showed a new approach for public-key quantum money using modular forms. Ji et al. \cite{ji2018pseudorandom} defined the pseudorandom quantum state (PRS) and gave a private-key quantum money scheme based on PRS. Recently, Peter Shor \cite{S20} proposed a public-key quantum money scheme based on the hardness of a lattice problem.

Another interesting circumstance to consider is classically verifiable quantum money introduced in \cite{gavinsky2014classicalqm}.  \cite{radian2019semi} gave a construction for semi-quantum money which can be verified with a protocol over classical channels.



\paragraph{One-time Programs and One-time Memory}
Another idea of copy-protecting softwares is through one-time program, introduced in \cite{goldwasser2008onetime}. One-time programs can be executed on only one single input and nothing other than the result of this computation is leaked. 
Quantum one-time programs are discussed in \cite{broadbent2013quantum}, showing that any quantum circuit can be compiled into a one-time program assuming
only the same basic one-time memory devices used for classical circuits. 
\cite{LSZ20one-time-memo} constructs
one-time programs from quantum-accessible one-time memories where
the view of an adversary, despite making quantum queries, can be simulated by making only classical queries to the ideal functionality.

\section{Preliminaries}
We use $\lambda$ as the security parameter and when inputted into an algorithm, $\lambda$ will be represented in unary.
We say a function $\epsilon(x)$ is \emph{negligible} if for all inverse polynomials $1/p(x)$, $\epsilon(x)<1/p(x)$ for all large enough $x$. We use $\negl(x)$ to denote a negligible function. 
We use \emph{QPT} to denote quantum polynomial time.

\subsection{Quantum Computation}
We give some basic definitions of quantum computation and quantum information in Appendix~\ref{sec:basic_qc}. Here, we only state a key Lemma for our construction: the Gentle Measurement Lemma proposed by Aaronson \cite{aaronson2004limitations}, which gives a way to perform measurements without totally destroying the state.

\begin{lemma}[Gentle Measurement Lemma \cite{aaronson2004limitations}] 
\label{lem:gentle_measure}
Suppose a measurement on a mixed state $\rho$ yields a
particular outcome with probability 
$1-\epsilon$.  Then after
the measurement, one can recover a state $\tilde{\rho}$ such that $ \left\lVert \tilde{\rho} - \rho \right\rVert_{\mathrm{tr}} \leq \sqrt{\epsilon}$.
\end{lemma}

\subsection{Quantum Oracle Algorithm}

In this work, we consider the quantum query model, which gives quantum circuits access to some oracles. 
\begin{definition}[Classical Oracle]
\label{def:classic_oracle}
A classical oracle $\mathcal{O}$ on input query $x$ is a unitary transformation of the form $U_f \ket{x,y, z} \rightarrow \ket{x, y+f(x), z}$ for classical function $f: \{0,1\}^n \rightarrow \{0,1\}^m$. Note that a classical oracle can be queried in quantum superposition.
\end{definition}
In the rest of the paper, the word `oracle' means a classical oracle. 
A quantum oracle algorithm with oracle access to $\mathcal{O}$ is a sequence of unitary $U_i$ and oracle access to $\mathcal{O}$ (or $U_f$). The query complexity of a quantum oracle algorithm is the number of $\mathcal{O}$ access.

In the analysis of security of the copy-protection scheme in Section~\ref{sec:security_proof}, we will use the theorem from \cite{BBBV1997} to bound the change in adversary's state when we change the oracle's input-output at where the adversary hardly ever queries on.

\begin{theorem}[\cite{BBBV1997}] 
\label{thm:bbbv97_oraclechange}
Let $\ket{\phi_i}$ be the superposition of quantum Turing machine $\cM$ with oracle $\cO$ on input $x$ at time $i$. Define $W_y(\ket{\phi_i})$ to be the sum of squared magnitudes in $\ket{\phi_i}$ of configurations of $\cM$ which are querying the oracle on string $y$. For $\epsilon  > 0$, let $F \subseteq [0, T-1] \times \Sigma^*$ be the set of time-string pairs such that 
$\sum_{(i,y) \in F} W_y(\ket{\phi_i}) \leq \epsilon^2/T$.

Now suppose the answer to each query $(i, y) \in F$ is modified to some arbitrary fixed $a_{i,y}$ (these answers need not be consistent with an oracle). Let $\ket{\phi_i'}$ be the superposition of $\cM$ on input $x$ at time $i$ with oracle $\cO$ modified as stated above. Then $\left\| \ket{\phi_T} - \ket{\phi_T'} \right\|_{\mathrm{tr}}\leq \epsilon$.
\end{theorem}

\subsection{Direct-Product Problem and Quantum Signature Tokens}
\label{prelim:directproduct_sigtoken}

In this section, we will define direct-product problem, which are key components of quantum signature token scheme by Ben-David and Sattath \cite{ben2016quantum} and also our quantum copy-protection scheme.

\begin{definition}[Dual Subspace] 
\label{def:dual_space}
Given a subspace $S$ of a vector space $V$, let $S^{\bot}$ be the
orthogonal complement of $S$: the set of $y \in V$ such that $x \cdot y = 0 $ for all $x \in S$. It is not
hard to show: $S^\bot$ is also a subspace of $V$; $(S^\bot)^\bot = S$.
\end{definition}


\begin{definition}[Subspace Membership Oracles]
A subspace membership oracle for a subspace $A \subseteq \F^n$, denoted as $U_A$, on input vector $v$, will output 1 if $v \in A$, $v \neq 0$ and  output 0 otherwise.
\end{definition}

\begin{definition}[Subspace State]
For a subspace $A \subseteq \F^n$, the state $\ket A$ is defined as $\frac{1}{\sqrt{|A|}} \sum_{v \in A} \ket v$, which is a uniform superposition of all vectors in $A$. 
\end{definition}

\paragraph{Direct-Product Problem} 
Our construction relies on the following problem
called the ``Direct-Product Problem" in \cite{aaronson2012quantum}: for any QPT adversary $\cA$, given one copy of $\ket{A}$ and oracle access to $U_A, U_{A^\perp}$, the problem is to finds two \textit{non-zero} vectors such that $u \in A$ and $v \in A^\perp$.

The hardness of the direct-product problem was proved by Ben-David and Sattath \cite{ben2016quantum}, used for construction of quantum signature tokens. More precisely, a signature token is  a subspace state $\ket A$ in their construction. All vectors in $A\setminus \{0\}$ are signatures for bit $0$ and all vectors in $A^\perp\setminus \{0\}$ are signatures for bit $1$. 
Therefore, to generate valid signatures for both 0 and 1, it is required to solve the ``Direct-Product Problem". Our copy-protection scheme works for general signature token schemes. To keep the statement and proof simple, we focus on the construction in \cite{ben2016quantum}.



\begin{theorem}[\cite{ben2016quantum}]
\label{thm:restate_bound_twopts}
Let $\epsilon>0$ be such that $1/\epsilon=o(2^{n/2})$. Let $A$ be a random subspace $\F^n$, and $\dim(A) =n/2$. Given one copy of $\ket{A}$ and access to subspace membership oracles of $U_A$ and $U_{A^\bot}$, an adversary needs $\Omega(\sqrt{\epsilon}2^{n/4})$ queries to output a pair of non-zero vectors $(u,v)$ such that $u\in A$ and $v\in A^\bot$ with probability at least $\epsilon$.
\end{theorem}

We will refer to the direct-product problem as a security game, which is defined as follows:
\begin{definition}[Direct-Product Game]
A direct-product game consists of the following steps:
\begin{description}
\item \textbf{Setup Phase}: the challenger takes in a security parameter $\lambda$, samples a random $\lambda/2$-dimensional subspace $A$ from $\F^\lambda$; then prepares the membership oracle $U_A$ for $A$, $U_{A^\perp}$ for the dual subspace $A^\perp$ and a quantum state $\ket{A}$.
\item \textbf{Query Phase}: the challenger sends $\ket{A}$ to the adversary; the adversary can query $U_A, U_{A^\perp}$ for polynomially many times. 
\item \textbf{Output Phase}: the adversary outputs two vectors $(u,v)$. 
\end{description}
The challenger checks if $u \in A\setminus\{0\}, v \in A^\perp \setminus \{0\}$. If this is satisfied, then the adversary wins.
\end{definition}

\cref{thm:restate_bound_twopts} shows that for any QPT adversary, the winning probability of the direct-product game is negligible.

\subsection{Measurement Implementation}

The following definitions and lemmas are introduced by Zhandry \cite{z20}. 

\begin{definition}[Controlled Projection]
\label{def:controlled_project}
    Let $\cP = \{\cP_i\}_{i \in \cal I}$ be a collection of projective measurement over a Hilbert space $\cal H$, where $\cP_i = (P_i, Q_i)$ for $i\in \cal I$. Let $D$ be a distribution with a random coin set $\cR$. 
    We define the controlled projection, denoted $\cproj_{\cP, D} = (\cproj^0_{\cP, D}, \cproj^1_{\cP, D})$ as the follows: 
    \begin{align*}
        \cproj^0_{\cP, D} := \sum_{r \in \cR} \ket r \bra r \otimes P_{D(r)} \quad\quad\quad \cproj^1_{\cP, D} := \sum_{r \in \cR} \ket r \bra r \otimes Q_{D(r)}
    \end{align*}
\end{definition}
In other words, $\cproj_{\cP, D}$ uses the random coins $r$ as a control and decides which projective measurement to be applied on the system. That is, $\cproj_{\cP, D}$ implements the following mixed projective measurement, which is a POVM $\cP_D = (P_D, Q_D)$ where
$P_D = \sum_{i \in \cal I} \Pr[i \gets D] \, P_i$ and  $Q_D = \sum_{i \in \cal I} \Pr[i \gets D] \, Q_i$. 

For example, $D$ generates a random message $m$ and a random encryption $c$ of this message $m$. 
In this case, ${\cal I} = \{(m, c)\}$ for all messages and ciphertexts.
$\cP_{(m, r)} = (P_{(m, r)}, Q_{(m,r)})$ simply means trying to decrypt a ciphertext $c$ and check if the resulting message is equal to $m$. 

\begin{definition}[Projective Implementation]
\label{def:project_implement}
    Let $\cP = (P, Q)$ be a binary outcome POVM. Let $\cD$ be a finite set of distributions over outcomes $\{0, 1\}$. Let $\cE = \{E_D\}_{D \in \cD}$ be a projective measurement with index set $\cD$. Consider the following measurement: 
    \begin{itemize}
        \item Measure under the projective measurement $\cE$ and obtain a distribution $D$ over $\{0, 1\}$;
        \item Output a bit according to the distribution $D$. 
    \end{itemize}
    We say the above measurement is a projective implementation of $\cP$ if it is equivalent of $\cP$, denoted as $\projimp(\cP)$. 
\end{definition}

Note that if the outcome is a distribution $D = (d_0, d_1)$, the collapsed state is an eigenvector of $P$ corresponding to eigenvalue $d_0$, and it is also an eigenvector of $Q$ corresponding to eigenvalue $d_1 = 1 - d_0$.

\begin{lemma}[A variation of Lemma 1 in \cite{z20}]
\label{lem:proj_implement}
    Any binary outcome POVM $\cP = (P, Q)$ has a projective measurement $\projimp(\cP)$.
\end{lemma}

In this work, we propose the following new definition corresponding to $\projimp$.
\begin{definition}[Threshold Implementation]
\label{def:thres_implement}
A threshold implementation with parameter $\gamma$ of a binary POVM $\cP = (P, Q)$ is a variant of projective implementation $\projimp(\cP)$, denoted as $(\ti_\gamma(\cP), {\bf I} - \ti_{\gamma}(\cP))$:  
\begin{itemize}
        \item Instead of measuring under the projective measurement $\cE = \{E_D\}_{D \in \cD}$ and obtain a distribution $D$ over $\{0, 1\}$, $\ti_{\gamma}(\cP)$ measures if the corresponding distribution $D = (d_0, d_1)$ has $d_0 \geq \gamma$. 
        \item Output $0$ with probability $\Tr[\ti_{\gamma}(\cP) \rho]$ and $1$ with probability $1- \Tr[\ti_{\gamma}(\cP) \rho]$, for any quantum state $\rho$.
\end{itemize}
Therefore, $\ti_{\gamma}(\cP)$ is a projection and the collapsed state is a (mixed) state in the span of all eigenvectors of $P$ whose eigenvalues are at least $\gamma$. 
\end{definition}
\begin{remark}
    For a binary outcome measurement $\cP = (P_0 ,P_1)$, we usually say `perform measurement $P_0$ on $\rho$' if $\cP$ was performed on $\rho$.  Since we only focus on the case that outcome is $0$ in the paper, it sometimes also denotes applying $\cP$ on $\rho$ conditioned on that the outcome is $0$.
\end{remark}

\subsubsection{Approximating Projective Implementation}

Before describing the theorem of the approximation algorithm, we give two definitions that characterize how good an approximation projective implementation is, which were first introduced in \cite{z20}. 

\begin{definition}[Shift Distance]
    For two distribution $D_0, D_1$, the shift distance with parameter $\epsilon$ is defined as $\Delta_{\shift}^\epsilon(D_0, D_1)$, which is the smallest quantity $\delta$ such that for all $x\in \R$: 
    \begin{align*}
        \Pr[D_0 \leq x] &\leq \Pr[D_1 \leq x + \epsilon] + \delta, \\
        \Pr[D_1 \leq x] &\leq \Pr[D_0 \leq x + \epsilon] + \delta.
    \end{align*}
    
    For two real-valued measurements $\cM$ and $\cN$ over the same quantum system, the shift distance between $\cM$ and $\cN$ with parameter $\epsilon$ is defined as,
    \begin{align*}
        \Delta^\epsilon_{\shift}(\cM, \cN) := \sup_{\ket \psi} \Delta_{\shift}^\epsilon\left( \cM (\ket \psi), \cN (\ket \psi) \right).
    \end{align*}
\end{definition}

\begin{definition}[$(\epsilon, \delta)$-Almost Projective]
    A real-valued quantum measurement $\cM$ is said to be $(\epsilon, \delta)$-almost projective if for all quantum state $\ket \psi$, apply $\cM$ twice in a row to $\ket \psi$, obtaining outcomes $X$ and $Y$. Then we have
        $\Pr[|X - Y| \leq \epsilon] \geq 1 - \delta$. 
\end{definition}

\begin{theorem}[Theorem 2 in \cite{z20}]\label{thm:api_projimp}
Let $D$ be any probability distribution and $\cP$ be a collection of projective measurements. For any $0<\epsilon, \delta < 1$, there exists an algorithm of measurement $\api^{\epsilon, \delta}_{\cP, D}$ that satisfies the followings:
    \begin{itemize}
        \item $\Delta^\epsilon_{\shift}(\api^{\epsilon, \delta}_{\cP, D}, \projimp(\cP_D)) \leq \delta$.
        \item $\api^{\epsilon, \delta}_{\cP, D}$ is $(\epsilon, \delta)$-almost projective. 
        \item The expected running time of $\api^{\epsilon, \delta}_{\cP, D}$ is $T_{\cP, D} \cdot \poly(1/\epsilon, \log(1/\delta))$ where $T_{\cP, D}$ is the combined running time of $D$, the procedure mapping $i$ to $(P_i, Q_i)$ and the run-time of measurement $(P_i, Q_i)$. 
    \end{itemize}
\end{theorem}

\section{Learning Game Definitions}
\label{sec:new_lg_defs}

\subsection{Unlearnability}

\begin{definition}[Quantum Program with Classical Inputs and Outputs]
A quantum program with classical inputs is a pair of quantum state $\rho$ and unitaries $\{U_x\}_{x \in [N]}$ (where $[N]$ is the domain), such that the state of the program evaluated on input $x$ is equal to $U_x \rho U^\dagger_x$.  To obtain an output, it measures the first register of $U_x \rho U^\dagger_x$. 
Moreover, $\{U_x\}_{x \in [N]}$ has a compact classical description which means applying $U_x$ can be efficiently computed given $x$. 
\end{definition}

Notation-wise, the input and output space $N,M$ are functions in $\lambda$.

\begin{definition}[$\gamma$-Goodness Test with respect to $f, D$] \label{def:gamma_good_test_functionality_equivalence}
    Let $(\rho_f, \{U_{f,x}\}_{x \in [N]})$ be a quantum program for computing a classical function $f:[N]\rightarrow [M]$. Let $D$ be a probability distribution over the input space $[N]$. 
    \begin{itemize}
        \item Define $(P_{f, x}, Q_{f, x})$ be a projective measurement that computes the quantum program on input $x$, and checks in superposition that if the quantum circuit outputs correctly. Let $V_{f, x}$ be a projection that checks if in superposition, the first register is equal to $f(x)$. We have $P_{f, x} = V_{f, x} U_{f, x}$ and $Q_{f, x} = {\bf I} - P_{f,x}$. 
        \item Let $\{P_f, Q_f\}$ be the controlled projection with respect to the  distribution $D$, as defined in \cref{def:controlled_project}. Then, let $\{\ti_{\gamma}(P_f), {\bf I} - \ti_{\gamma}(P_f) \}$ be the Threshold Implementation for $P_f$ with threshold value $\gamma$, as defined in \cref{def:thres_implement}.
        
        \item We say a quantum program is tested \textbf{$\gamma$-good for computing $f$ with distribution $D$} if the projective measurement $\{\ti_{\gamma}(P_f), {\bf I} - \ti_{\gamma}(P_f) \}$ on $\rho_f$ outputs 0.
    \end{itemize}
\end{definition}

\begin{definition}[Learning Game for $\cF, \cD$] \label{def:learning_game_functional_equivalent}
A learning game for a function family $\cF = \{\cF_\lambda: [N] \to [M]\}$ , a distribution family $\cD = \{D_f\}$, 
and an adversary $\As$ is denoted as $\LG^\As_{\mathcal{F}, \cD, \gamma}(1^\lambda)$, which consists the following steps:
\begin{enumerate}
    \item \textbf{Sampling Phase}: At the beginning of the game, the challenger takes a security parameter $\lambda$ and samples a function $f \gets \cF_\lambda$; 
    \item \textbf{Query Phase}: $\As$ then gets oracle access to $f$;
    \item \textbf{Output Phase}: Finally, $\As$ outputs a quantum program $(\rho, \{U_x\}_{x \in [N]})$.
\end{enumerate}
The game outputs $0$ if and only if the program is tested to be $\gamma$-good with respect to $f, D_f$.
\end{definition}

\begin{definition}[Quantum Unlearnability of $\cF$ with Testing Distribution $\cD$] \label{def:unlearnable_func}
A family of functions with respect to $\cD$ is called $\gamma$ quantum unlearnable if for all $\lambda$, for any QPT adversary $\cA$, there exists a negligible function $\negl(\cdot)$ such that the following holds:
 \begin{align*}
       \Pr\left[ b = 0,\, b \gets \LG_{\cF,\cD, \gamma}^\cA(1^\lambda) \right]  \leq \negl(\lambda)
    \end{align*}
\end{definition}

\subsection{Generalized Unlearnability}

The $\gamma$-goodness test for quantum program (\cref{def:gamma_good_test_functionality_equivalence}) captures the intuition that a quantum program's behavior on classical inputs is $\gamma$-good comparing to the input-output behavior of $f$ with respect to the input distribution $D_f$. For cryptographic primitives, as discussed in the introduction, achieving a particular cryptographic functionality does not necessarily mean to have the exact input-output behavior.  As an example, to sign a message, there are usually more than one valid signatures and the intended functionality is preserved as long as any valid signature is provided. 

For a randomized function $f$, we denote the input $x$ of $f$ as the real input taken by $f$ as well as random coins used by $f$. 
\begin{definition}[Predicate]
    A classical predicate $E(P, y_1, \cdots, y_k, r)$ is a binary outcome function that runs a classical program $P$ on a randomly sampled input $x$ to get output $z$, and outputs $0/1$ depending on whether $(z, y_1, \cdots, y_k, r) \in R$ for some binary relation defined by $R$. The randomness of input $x$, program $P$ all depends on randomness $r$. $y_1, \cdots, y_k$ are auxiliary inputs that specify the relation. 
    
    Quantumly, it runs a quantum program on random classical input $x$ and measure if $(z, y_1, \cdots, y_n, r) \in R$ in superposition, where $z$ is the first register of the resulting state. In other words, it is a projective measurement indexed by $r$. 
\end{definition}

We use $\samp, \mathscr{F}$ to denote a cryptographic application. $\mathscr{F}$ denotes the intended functionality that this cryptographic application should achieve. 
\begin{definition}[Cryptographic Application $\samp, \mathscr{F}$]
$\samp$ is a sampler that takes a security parameter $\lambda$ and interacts with an adversary $\As$: $f \gets (\As \Leftrightarrow \samp(1^\lambda))$ where $f$ is a classical circuit that contains some secret information $s_f$ which is unknown to $\As$, and $\As$ can get some public information $\aux_f$ from the interaction. 

$\mathscr{F} = \{F_\lambda\}$ and $F_\lambda(P, f, r)$ is a predicate which takes a program, a circuit $f$ and randomness $r$.  
For all efficient $\As$, all $ f$ sampled by $\samp$, there exists a negligible function $\negl(\cdot)$ such that, $\Pr\left[  F_\lambda(f, f, r) = 0 \right] \geq 1 - \negl(\lambda)$.
\end{definition}

This security of the cryptographic application is orthogonal to its correctness and unlearnability. The definition of security varies a lot when different applications are given. 
Some examples include CPA security for public key encryption schemes and signature unforgeability. However,  the security should be easy to prove, when we implement a copy protection/copy detection scheme using our construction. In this paper, we only focus on its correctness and copy-protect security/copy-detect security/unlearnability/unremovability.

\begin{definition}[$\gamma$-Goodness Test with respect to $f, E$] \label{def:gamma_good_test_crypto_applications}
    Let a quantum program  for computing $f$ be $(\rho_f, \{U_{f,x}\}_{x \in [N]})$. 
    \begin{itemize} 
        
        \item Quantumly, define $(P_{f, r}, Q_{f, r})$ be a projective measurement that computes the quantum program on input $x$ (sampled according to $r$), and checks in superposition that if the output of the quantum circuit satisfies the predicate $E(\cdot, f, r)$ in superposition.
        \item Let $\{P_f, Q_f\}$ be the controlled projection with respect to uniform distribution on randomness $r$. Let $\{\ti_{\gamma}(P_f), {\bf I} - \ti_{\gamma}(P_f) \}$ be the threshold implementation for $P_f$ with threshold value $\gamma$.
        
        \item A quantum program is \textbf{tested $\gamma$-good} with respect to $f,E$ if the projective measurement $\{\ti_{\gamma}(P_f), {\bf I} - \ti_{\gamma}(P_f) \}$ on $\rho_f$ outputs 0.
    \end{itemize}
\end{definition}

Note that \cref{def:gamma_good_test_functionality_equivalence} fits into this general definition, where the predicate $E$ on a random input $x$ ($x$ is drawn depending on randomness $r$) and $f$, checks if the output is equal to $f(x)$. 

We then generalize the learning game to the setting of cryptographic applications.  Note that $\mathscr{E}$ may be not the same as $\mathscr{F}$. In the game below, an adversary tries to learn a more restricted functionality of $f$. 
\begin{definition}[Learning Game for $\samp, \mathscr{E}$] \label{def:learning_game_crypto_applications}
A learning game for a sampler $\samp$ (which samples a function in $\mathcal{F}_\lambda$), a predicate $\mathscr{E} = \{E_\lambda\}$,
and an adversary $\As$ is denoted as $\LG^\As_{\samp, \mathscr{E}, \gamma}(1^\lambda)$, which consists the following steps:
\begin{enumerate}
    \item \textbf{Sampling Phase}: At the beginning of the game, $\As$ interacts with the challenger and samples $f \gets (\As \Longleftrightarrow \samp(1^\lambda))$. 
    \item \textbf{Query Phase}: $\As$ then gets oracle access to $f$;
    \item \textbf{Output Phase}: Finally, $\As$ outputs a quantum program $(\rho, \{U_x\}_{x \in [N]})$. 
\end{enumerate}
The game outputs $0$ if and only if the program is tested to be $\gamma$-good with respect to $f, E_\lambda$. 
\end{definition}

It is easy to see that \cref{def:learning_game_crypto_applications} implies \cref{def:learning_game_functional_equivalent}.  One example is digital signature. $\samp$ picks a pair of signing key and verification key $(\sk, \vk)$ and outputs a signing circuit $f = \sign(\sk, \cdot)$ which hard-wires $\sk$ and appends $\vk$ with the circuit description. The predicate is defined as: sample $m, r_s, r_v$ according to randomness $r$,  run the program with input $m$ and randomness $r_s$ to obtain outcome $z$, 
decode $\sk, \vk$ from the circuit $f$ and the predicate is $0$ if and only if $\ver(\vk, m, z; r_v) = 1$. In other words, the predicate checks if the program outputs a valid signature on a random message. 


\begin{definition}[Quantum Unlearnability of $(\samp, \mathscr{F}), \mathscr{E}$]
$((\samp, \mathscr{F}), \mathscr{E})$ is called $\gamma$-quantum-unlearnable if for all $\lambda$, for any QPT adversary $\cA$, there exists a negligible function $\negl(\cdot)$ such that the following holds: 
 \begin{align*}
       \Pr\left[ b = 0,\, b \gets \LG_{\samp,\mathscr{E}, \gamma}^\cA(1^\lambda) \right]  \leq \negl(\lambda)
    \end{align*}
\end{definition}

\subsection{Generalized Copy Protection}

\begin{definition}[Quantum Copy Protection]
\label{def:qcp_general}
A quantum copy-protection scheme for
$(\samp, \mathscr{F}), \mathscr{E}$ consists of the following procedures:
\begin{description}
    \item[]\setup$(1^\lambda) \to (\sk)$: the setup algorithm takes in a security parameter $\lambda$ in unary and generates a secret key $\sk$. 
    
    \item[] \generate$(\sk, f) \rightarrow (\rho_f, \{U_{f, x}\}_{x \in [N]})$:
    on input $f \in \cF_\lambda$ and secret key $\sk$, the vendor generates a quantum program $(\rho_f, \{U_{f, x}\}_{x \in [N]})$. 

 \item[]\compute$ (\rho_f, \{U_{f, x}\}_{x \in [N]}, x) \rightarrow y $:
     given a quantum program, a user can compute the function $f(x)$ on input $x$ by applying $U_{f, x}$ on $\rho$ and measuring the first register of the state. 
\end{description}
\end{definition}

\begin{description}
 \item[] \textbf{Efficiency:}  \textsf{Setup}, $\compute$ and \textsf{Generate}  should run in  $\poly(\lambda)$ time. 
 
    \item[] \textbf{Correctness:}  
     For all $\lambda \in \N$, all efficient $\As$,  every $f \gets (\As \Longleftrightarrow \samp(1^\lambda))$,
     all $(\rho_f, \{U_{f, x}\}_{x \in [N]}) \gets \generate(\sk, f)$, 
     there exists a negligible function $\neglfunc(\cdot)$ such that,
     \begin{description}
         \item \textbf{unique output}: for all $x \in [N]$, apply $U_{f, x}$ on $\rho_f$ and measure the first register, with probability at least $1-\negl(\lambda)$, the output is a fixed value $z_{f, x}$;
         
         
         \item \textbf{functionality preserving}:
         $(\rho_f, \{U_{f, x}\}_{x \in [N]})$ are $(1-\negl(\lambda))$-good with respect to $f, F_\lambda$ with probability $1$. 
     \end{description}
     
    
\item[] \textbf{Security}: It has $\gamma$-anti-piracy security defined below. 
\end{description}
Note that the property ``unique output'' enables the copy-protected program can be evaluated polynomially many times.

\begin{definition}[$\gamma$-Anti-Piracy Security Game]
An anti-piracy security game for a sampler $\samp$, a predicate $\mathscr{E}$ and adversary $\cA$ is denoted as $\AG_{\samp, \mathscr{E}, \gamma}^{\cA}(1^\lambda)$, which consists of the following steps:

\begin{enumerate}
    \item \textbf{Setup Phase}: At the beginning of the game, the challenger takes a security parameter $\lambda$ and obtains secret key $\sk \gets \setup(1^\lambda)$.
    
    \item \textbf{Sampling Phase}: $\As$ interacts with the challenger and samples $f \gets (\As \Longleftrightarrow \samp(1^\lambda))$.
    
    \item \textbf{Query Phase}:
    $\As$ makes a single query to the challenger and obtains a copy protection program for $f$: $(\rho_f, \{U_{f, x}\}_{x \in [N]}) \gets \generate(\sk, f)$. 

    \item \textbf{Output Phase}: Finally, $\As$ outputs a (possibly mixed and entangled) state $\sigma$ over two registers $R_1, R_2$ and two sets of unitaries $(\{U_{R_1, x}\}_{x \in [N]}, \{U_{R_2, x}\}_{x \in [N]})$  They can be viewed as programs $\P_1 = (\sigma[R_1], \{U_{R_1, x}\}_{x \in [N]})$ and $\P_2 = (\sigma[R_2],\allowbreak \{U_{R_2, x}\}_{x \in [N]})$.
\end{enumerate}
The game outputs $0$ if and only if \emph{both} programs $\P_1, \P_2$ are both tested to be $\gamma$-good with respect to $ E_\lambda$. 

\end{definition}

Similarly, we can define $q$-collusion resistant $\gamma$-anti-piracy security game $\AG_{\samp, \mathscr{E}, \gamma}^{q, \cA}(1^\lambda)$, in which the adversary $\cA$ can make  at most $q$ queries in the query phases and is required to output $q+1$ programs $\{(\sigma[R_i],\{U_{R_i,x}\}_{x\in [N]})\}_{i\in [q+1]}$ such that each program is tested to be $\gamma$-good.

\begin{definition}[$\gamma$-Anti-Piracy-Security] 
 A copy protection scheme for $\samp$ and $\mathscr{E}$ has $\gamma$-anti-piracy security, if for any QPT adversary $\cA$,  there exists a negligible function $\negl(\cdot)$ such that the following holds for all $\lambda \in \N$: 
  \begin{align}
    \Pr\left[b = 0, b \gets \AG_{\samp, \mathscr{E}, \gamma}^{\cA}(1^\lambda) \right]\leq \negl(\lambda)
    \end{align}
\end{definition}


\subsection{Generalized Copy Detection}

A copy detection scheme for $(\samp, \mathscr{F}), \mathscr{E}$ is very similar to the copy protection scheme, except it has an additional procedure $\chk$ which applies a projective measurement and checks if the quantum state is valid. 

\begin{definition}[Quantum Copy Detection]
\label{def:qcd_general}
A quantum copy-detection scheme for
$(\samp, \mathscr{F}), \mathscr{E}$ consists of the following procedures:
\begin{description}
     \item[]\setup$(1^\lambda)$, \generate$(\sk, f)$ and \compute$ (\rho_f,  \{U_{f, x}\}_{x \in [N]}, x)$ are the same as those in \cref{def:qcp_general}.

    


\item[]\chk$(\pk, \aux_f, \rho_f, \{U_{f, x}\}_{x \in [N]}) \to b, \rho'$: on input a public key $\pk$, public information $\aux_f$ generated during $\samp$, a quantum program, it applies a binary projective measurement $P_0, P_1$ on $\rho_f$ that depends on $\pk$, $\aux_f$, $\{U_{f, x}\}_{x \in [N]}$; it outputs the outcome $b$ and the collapsed state $\rho'$.
\end{description}
\end{definition}

\begin{description}
 
    \item[] \textbf{Correctness} (\generate):  The same as the security of \cref{def:qcp_general}. 
    
    \item[] \textbf{Correctness} (\chk):  
     For all $\lambda \in \N$, all efficient $\As$,  every $f \gets (\As \Longleftrightarrow \samp(1^\lambda))$,
     all $(\rho_f, \{U_{f, x}\}_{x \in [N]}) \gets \generate(\sk, f)$, 
     there exists a negligible function $\neglfunc(\cdot)$ such that, $\chk(\pk,  \aux_f, \rho_f, \{U_{f, x}\}_{x \in [N]})$ outputs $0$ with probability at least $1 - \negl(\lambda)$. 
    
\item[] \textbf{Security}: It has $\gamma$-copy-detection security defined below. 
\end{description}

\begin{definition}[$\gamma$-Copy-Detection Security Game]
A copy-detection security game for a sampler $\samp$, a predicate $\mathscr{E}$ and adversary $\cA$ is denoted as $\DG_{\samp, \mathscr{E}, \gamma}^\cA(1^\lambda)$, which consists of the following steps:

\begin{enumerate}
    \item \textbf{Setup Phase}: At the beginning of the game, the challenger takes a security parameter $\lambda$ and obtains keys $(\pk, \sk) \gets \setup(1^\lambda)$.
    
    \item \textbf{Sampling Phase}: $\As$ interacts with the challenger and samples $f \gets (\As \Longleftrightarrow \samp(1^\lambda))$. Let $\aux_f$ denote the public information $\As$ obtains during the interaction.

    \item \textbf{Query Phase}:
    $\As$ makes a single query to the challenger and obtains a copy detection program for $f$: $(\rho_f, \{U_{f, x}\}_{x \in [N]}) \gets \generate(\sk, f)$. 
    
    \item \textbf{Output Phase}: Finally, $\As$ outputs a state $\sigma$ over two registers $R_1, R_2$ and two sets of unitaries $(\{U_{R_1, x}\}_{x \in [N]}, \{U_{R_2, x}\}_{x \in [N]})$. They can be viewed as programs $\P_1 = (\sigma[R_1], \{U_{R_1, x}\}_{x \in [N]})$ and $\P_2 = (\sigma[R_2], \{U_{R_2, x}\}_{x \in [N]})$.
\end{enumerate}
The game outputs $0$ if and only if
\begin{itemize}
    \item Apply $\chk$ on input $\pk, \aux_f, P_i$ respectively and both outcomes are $0$. Let $P'_i$ be the collapsed program conditioned on outcomes are $0$. 
    \item \emph{Both} programs $\P'_1, \P'_2$ are both tested to be $\gamma$-good with respect to $f, E_\lambda$. 
\end{itemize}
\end{definition}

Similarly, we can define $q$-collusion resistant $\gamma$-copy-detection security game $\DG_{\samp, \mathscr{E}, \gamma}^{\cA, q}(1^\lambda)$, in which the adversary $\cA$ can perform at most $q$ query phases and output $q+1$ programs $P_i=(\sigma[R_i], \{U_{R_i, x}\}_{x\in [N]})$ for $i\in [q+1]$. The game outputs 0 if and only if for all $i\in [q+1]$, the outcome of applying $\chk$ on $P_i$ is $0$, and the collapsed program $P_i'$ is tested to be $\gamma$-good. 

\begin{definition}[$\gamma$-Copy-Detection-Security]
 A copy detection scheme for $\samp$ and $\mathscr{E}$ has $\gamma$-security, if for any QPT adversary $\cA$,  there exists a negligible function $\negl(\cdot)$ such that the following holds for all $\lambda \in \N$: 
  \begin{align}
    \Pr\left[b = 0, b \gets \DG_{\samp, \mathscr{E}, \gamma}^{\cA}(1^\lambda) \right]\leq \negl(\lambda)
    \end{align}
\end{definition}

\subsection{Watermarking Primitives with Public Extraction}

In this subsection, we give a unified definition that covers most of the definitions in the previous works about watermarking primitives including \cite{cohen2018watermarking,kim2017watermarking,quach2018watermarking,kim2019watermarking,goyal2019watermarking}.   
We will give several concrete examples of watermarking schemes in Appendix~\ref{sec:wmapp}. 

\begin{definition}[Watermarking Primitives for $(\samp, \mathscr{F}), \mathscr{E}$]
A watermarking scheme for $(\samp, \mathscr{F}), \mathscr{E}$ consists of the following {classical} algorithms: 
\begin{description}
    \item \textbf{$\setup(1^\lambda)$}: it takes as input a security parameter $1^\lambda$ and outputs keys $(\xk, \mk)$. $\xk$ is the extracting key and $\mk$ is the marking key. 
    We only consider publicly extractable watermarking scheme. Thus $\xk$ is always public. 

    \item $\samp(1^\lambda)$: it takes a security parameter $1^\lambda$,
    \begin{align*}
        f \gets  (\As \Longleftrightarrow \samp(1^\lambda)). 
    \end{align*}
    We also denote $\aux_f$ as the public information $\As$ obtains during the interaction. 

    \item $\mar(\mk, f, \tau)$: it takes a circuit $f$ and a message $\tau \in \mathcal{M}_\lambda$, outputs a marked circuit $\widetilde{f}$. 
    
    \item $\extract(\xk, \aux_f, f')$: it takes the public auxiliary information $\aux_f$, a circuit and outputs a message in $\{\sf \bot\} \cup \mathcal{M}_\lambda$. 
    
\end{description}
\end{definition}

\emph{Remark.} In some watermarking schemes, $\setup$ also outputs a watermarking public parameter $\wpp$ and $\samp$ takes this parameter to sample a function. Our construction works in this setting. In sake of clarity, we use the above notion. $\extract$ may also take an $\aux$ that specifies its restricted functionality that $f'$ should achieve. We assume 
$f'$ contains a piece of information $\aux$ as a comment.

\vspace{1em}

It satisfies the following properties. 

\begin{definition}[Correctness of \mar\, (Functionality Preserving)]
    For all $\lambda$, for every efficient algorithm $\As$, there exists a negligible function $\negl$, for all $(\xk, \mk)\gets \setup(1^\lambda)$, and every $\tau \in \mathcal{M}_\lambda$, 
    \begin{align*}
        \Pr\left[ F_\lambda(\widetilde{f}, f, r) = 0 \,:\,   \substack{  f \gets (\As \Longleftrightarrow \samp(1^\lambda))  \\ \widetilde{f} \gets\mar(\mk, f, \tau)  }\right] \geq 1 - \negl(\lambda).
    \end{align*}
\end{definition}

\begin{definition}[Correctness of \extract]
    For all $\lambda$, for every efficient algorithm $\As$, there exists a negligible function $\negl(\cdot)$, for all $( \xk, \mk)\gets \setup(1^\lambda)$, and every $\tau \in \mathcal{M}_\lambda$, every $\aux$, 
    \begin{align*}
        \Pr\left[ \tau \ne \extract(\xk,  \aux_f, \widetilde{f} || \aux) \,:\,  \substack{ f \gets (\As \Longleftrightarrow \samp(1^\lambda)) \\ \widetilde{f} \gets\mar(\mk, f, \tau)  } \right] \leq \negl(\lambda),
    \end{align*} 
    where $\aux_f$ is the public information given to $\As$ and $\widetilde{f} || \aux$ is the program appended with $\aux$. 
\end{definition}

\begin{definition}[Meaningfulness]
    For all $\lambda$, for every efficient algorithm $\As$, there exists a negligible function $\negl(\cdot)$, for every $\aux$, 
    \begin{align*}
        \Pr\left[  \bot \ne \extract(\xk,  \aux_f, f||\aux) \,:\,  \substack{ (\xk, \mk)\gets \setup(1^\lambda) \\ f \gets (\As \Longleftrightarrow \samp(1^\lambda)) } \right] \leq \negl(\lambda).
    \end{align*}  
    where $\aux_f$ is the public information given to $\As$ and $\widetilde{f} || \aux$ is the program appended with $\aux$. 
\end{definition}

\begin{definition}[$\gamma$-Unremovability with respect to $\samp, \mathscr{E}$]
Consider the following game, denoted as $\WG^{\As}_{\samp, \mathscr{E}, \gamma}$: 
    \begin{enumerate}
        \item \textbf{Setup}: The challenger samples $(\xk, \mk) \gets \setup(1^\lambda)$. $\As$ then gets $\xk$. 
        \item \textbf{Sampling Phase}:  The challenger interacts with the algorithm $\As$ and samples $f \gets (\As \Longleftrightarrow \samp(1^\lambda))$. 
        \item \textbf{Query Phase}: $\As$ has classical access to $\mar(\mk, f, \cdot)$ at any time. Define $Q$ be the set of messages that $\As$ has queried on.
        
        \item \textbf{Output Phase}: Finally, the algorithm outputs a circuit $f^*$. 
    \end{enumerate}
    The adversary wins the game if and only if
    \begin{align*}
        \extract(\xk, \aux_f,  f^*) \not\in Q  ~\wedge~ \Pr_r[E_\lambda(f^*, f, r) = 1] \geq \gamma 
    \end{align*}
    We say a watermarking scheme has $\gamma$-unremovability respect to $\samp, \mathscr{E}$, if for all QPT $\As$, it wins the above game with negligible probability in $\lambda$. We say it has $q$-collusion resistant  $\gamma$-unremovability if the number of queries made in the query phase is at most $q$. 
\end{definition}

\section{Approximating Threshold Implementation}
\label{sec:ati}

By applying $\api^{\epsilon, \delta}_{\cP, D}$ and checking if the outcome is greater than or smaller than $\gamma$, we get a approximated threshold implementation $\ati^{\epsilon, \delta}_{\cP, D, \gamma}$. Here, we use $(\ati^{\epsilon, \delta}_{\cP, D, \gamma}, {\bf I}-\ati^{\epsilon, \delta}_{\cP, D, \gamma})$ to denote this binary POVM.

\cref{thm:api_projimp} gives the following corollary on approximating threshold implementation:

\begin{corollary}\label{cor:ati_thresimp}
    For any $\epsilon, \delta, \gamma, \cP, D$, the algorithm of measurement $\ati^{\epsilon, \delta}_{\cP, D, \gamma}$ that satisfies the followings: 
    \begin{itemize}
        \item For all quantum state $\rho$, $\Tr[\ati^{\epsilon, \delta}_{\cP, D, \gamma-\epsilon}\cdot \rho]\geq \Tr[\ti_\gamma(\cP_D)\cdot \rho]-\delta$.
        
        \item By symmetry, for all quantum state $\rho$, $\Tr[\ti_{\gamma-\epsilon}(\cP_D)\cdot \rho] \geq \Tr[\ati^{\epsilon, \delta}_{\cP, D, \gamma}\cdot \rho] -\delta$.
        
        \item 
        For all quantum state $\rho$, 
        let $\rho'$ be the collapsed state after applying  $\ati^{\epsilon, \delta}_{\cP, D, \gamma}$ on $\rho$. Then, $\Tr[\ti_{\gamma-2 \epsilon}(\P_D)\cdot \rho'] \geq 1 - 2\delta$. 
        \item
        The expected running time is the same as $\api_{\cP, D}^{\epsilon, \delta}$.
    \end{itemize}
\end{corollary}

Intuitively the corollary says that if a quantum state $\rho$ has weight $p$ on eigenvectors with eigenvalues at least $\gamma$, the measurement $\ati^{\epsilon, \delta}_{\cP, D, \gamma - \epsilon}$ with probability at least $p - \delta$ outputs outcome $0$ and the collapsed state has weight $1 - 2 \delta$ on eigenvectors with eigenvalues at least $\gamma - 2 \epsilon$.  Also note that the running time is proportional to $\poly(1/\epsilon, 1/(\log \delta))$, which is a polynomial in $\lambda$ as long as $\epsilon$ is any inverse polynomial and $\delta$ is any inverse sub-exponential function. The proof of the above Corollary is in Appendix~\ref{sec:proof_ati_th}.

We can also consider approximating the measurements on bipartite (possibly entangled) quantum state. We will prove a similar statement as \cref{cor:ati_thresimp}.
\begin{lemma} \label{lem:ati_2d}
Let $\cP_1$ and $\cP_2$ be two  collections of projective measurements and $D_1$ and $D_2$ be any probability distributions defined on the index set of $\cP_1$ and $\cP_2$ respectively. 
For any $0<\epsilon, \delta, \gamma < 1$, the algorithms $\ati^{\epsilon, \delta}_{\cP_1, D_1, \gamma}$ and $\ati^{\epsilon, \delta}_{\cP_2, D_2, \gamma}$ satisfy the followings:
\begin{itemize}
    \item For any bipartite (possibly entangled, mixed) quantum state $\rho\in \mathscr{H}_{\cal L}\otimes \mathscr{H}_{\cal R}$,
    \begin{align*}
        \Tr\big[\big(\ati_{\cP_1, D_1, \gamma-\epsilon}^{\epsilon, \delta}\otimes \ati_{\cP_2, D_2, \gamma-\epsilon}^{\epsilon, \delta}\big)\rho\big] \geq \Tr\big[ \big(\ti_{\gamma}(\cP_{D_1})\otimes \ti_\gamma(\cP_{D_2})\big)\rho \big] - 2\delta.
    \end{align*}
    \item For any (possibly entangled, mixed) quantum state $\rho$, 
        let $\rho'$ be the collapsed state after applying  $\ati^{\epsilon, \delta}_{\cP_1, D_1, \gamma}\otimes \ati^{\epsilon, \delta}_{\cP_2, D_2, \gamma}$ on $\rho$  (and normalized). Then, 
        \begin{align*}
            \Tr\big[ \big(\ti_{\gamma-2\epsilon}(\cP_{D_1})\otimes \ti_{\gamma-2\epsilon}(\cP_{D_2})\big) \rho'\big]\geq  1- 4 \delta.
        \end{align*}
\end{itemize}
\end{lemma}

We defer the proof of the above Lemma to Appendix~\ref{sec:proof_ati_2d}.

\section{Quantum Copy-Protection Scheme}



Let $\lambda$ be the security parameter. Let $\cF = \{\cF_\lambda\}_{\lambda \in \N}$ be a class of circuits. We assume $\cF$ is quantum unlearnable with respect to $\mathcal{D}$ and can be computed by polynomial-sized classical circuits. 
The construction for quantum copy-protection of function class $\cF_\lambda$
is defined in \cref{fig:cp_mini_scheme}.
\begin{figure}
    \centering
    \begin{gamespec}
    \begin{description}
    \item[]\textsf{Setup($1^\lambda ) \to \sk$:} The setup algorithm takes in security parameter $1^\lambda$.
    \begin{itemize}
        \item  Pick a uniformly random subspace $A\subseteq \F^\lambda$ of dimension $\lambda/2$.
        
        \item Output $\sk = A$, where $A$ is described by a set of orthogonal basis vectors.
    \end{itemize}
    
    \item[]\generate($\sk, f \in \cF_\lambda)$: The $\generate$ algorithm receives $\sk = A$ and a function $f$ from $\cF_\lambda$. 
    \begin{itemize}
        \item  Prepare a subspace state on $n$ qubits corresponding to $A$, $\ket A = \frac{1}{\sqrt{|A|}} \sum_{v \in A} \ket v$.  
    
    \item  Generate oracles $U_A, U_{A^\perp}$ which compute subspace membership functions for subspace $A$ and its dual subspace $A^\perp$ respectively.
    
    \item Generate oracles $\cO_1, \cO_2$ such that
    \begin{align*}
    \cO_1(x,v)&=\begin{cases}
    f(x)\oplus g(x)&\text{if }v\in A\text{ and }v\ne 0,\\
    \bot &\text{otherwise}.
    \end{cases}\\
    \cO_2(x,v)&=\begin{cases}
    g(x)&\text{if }v\in A^\bot\text{ and }v\ne 0,\\
    \bot &\text{otherwise}.
     \end{cases}
    \end{align*}
    where $g$ is a uniformly random function, with the same input and output length as $f$.  
    
    
    \item Finally, the $\generate$ algorithm outputs $\rho = \ket A \bra A$ and $\{U_x\}_{x \in [N]}$ describes the following procedure: 
        \begin{itemize}
            \item On input $x$, prepare the state $\ket 0 \bra 0 \otimes \ket x \bra x \otimes \rho$ and make an oracle query $U_A$ and measure the first register (output register) to get $y_1$; the remaining state is $\ket x \bra x \otimes \rho'$. 
            \item Apply ${\sf QFT}$ on the third register $\rho'$ to get $\rho''$.
            \item Prepare the state $\ket 0 \bra 0 \otimes \ket x \bra x \otimes \rho''$ and make an oracle query $U_{A^\perp}$ and measure the first register to get $y_2$. 
            \item Output $y_1 \oplus y_2$. 
        \end{itemize}
    \end{itemize}
    The description of $\{U_x\}_{x \in [N]}$ requires the oracle of $U_A, U_{A^\perp}$ (or the VBB obfuscations). 
    \end{description}
    \end{gamespec}
    \caption{Quantum copy-protection scheme.}
    \label{fig:cp_mini_scheme}
\end{figure}

Note that this construction works for general quantum unlearnable function families as well. By simply changing the notation in the proof to that in the general quantum unlearnability case, we prove it for general quantum unlearnable function families. More discussion will be given at the end of this section.

\paragraph{Oracle Heuristics} In practice we use a quantum-secure PRF \cite{zhandry2012construct} to implement function $g$; and we use quantum-secure (classical) VBB obfuscation to implement each of $(\cO_1,\cO_2,U_A, U_{A^\perp})$. We can replace VBB obfuscation programs with oracles that only allow black-box access by the security of VBB obfuscation; afterwards, we can also replace PRF $g$ with a real random function by the property of PRF. The heuristic analysis is straightforward and we omit them here.  





\subsection{Correctness and Efficiency}
\paragraph{Correctness} 
Given $\rho = \ket A \bra A$ and $\{U_x\}_{x \in [N]}$, it performs the following computation:
\begin{enumerate}
    \item Make an oracle query $\mathcal{O}_1$ on the state $\ket 0  \ket x \ket A$, the resulting state is statistically close to $\ket {y_1} \ket x \ket A$. Note that $\ket A$ with overwhelming probability $1 - 1/|A|$ contains a non-zero vector in $A$. It measures $y_1$, which is $y_1 = f(x) \oplus g(x)$. 
    \item It then prepares a state by applying QFT on the third register and the resulting state is is statistically close to $\ket 0 \ket x \ket {A^\perp}$. It makes an oracle query $\mathcal{O}_2$ on the state $\ket 0  \ket x \ket {A^\perp}$, the resulting state is statistically close to $\ket {y_2} \ket x \ket {A^\perp}$ where $y_2 = g(x)$. 
\end{enumerate}
Therefore, with overwhelming probability, the output is $y_1 \oplus y_2 = f(x)$.


\paragraph{Efficiency}
In $\generate$ algorithm, as shown in \cite{aaronson2012quantum}, given the basis of $A$, the subspace state $\ket A$ can be prepared in polynomial time using QFT. For the oracles $\cO_1, \cO_2$, it only needs to check the  membership of $A$ and $A^\bot$ and compute functions $f$ and $g$. $f$ can be prepared in polynomial time by definition. As we discussed above, we can prepare function $g$ as a PRF. 
Therefore, the oracles $\cO_1, \cO_2$ can be generated in polynomial time. 
The \textsf{Compute} algorithm is clearly efficient.

\subsection{Anti-Piracy Security}
\label{sec:security_proof}

We show that  for a \textit{quantum unlearnable} families of functions $\cF$ with respect to $\cD$ defined in \cref{def:unlearnable_func}, the quantum copy-protection scheme has anti-piracy security against any quantum polynomial-time adversaries. More formally: 

\begin{theorem}[Main Theorem]
    Let $\cF$ be a function families that is $\gamma$-quantum-unlearnable respect to distribution $\cD$ ($\gamma$ is a non-negligible function of $\lambda$). The above copy protection scheme for $\cF, \cD$ has $(\gamma(\lambda)-1/\poly(\lambda))$-anti-piracy security, for all polynomial $\poly$. 
\end{theorem}

In order to describe the quantum query behavior of quantum programs made to oracles, we give the following definitions and notations.

We recall that in \cref{def:gamma_good_test_functionality_equivalence}, a QPT adversary $\cA$ in the anti-piracy security game $\AG_{\cF,\cD,\gamma}^\cA(1^\lambda)$, will produce a state $\sigma$ over registers $R_1, R_2$ and unitaries $\{U_{R_1,x}\}_{x \in [N]}, \{U_{R_2,x}\}_{x \in [N]}$, the challenger  will then perform $\gamma$-goodness test on $\sigma$ using threshold implementations $\ti_{\gamma}(P_{R_1,f})$ and $\ti_{\gamma}(P_{R_2,f})$. 
For simplicity we will describe the unitary ensembles $\{U_{R_1,x}\}_{x \in [N]}$, $\{U_{R_2,x}\}_{x \in [N]}$  as
$U_{R_1}$, $ U_{R_2}$ and describe threshold implementations $\ti_{\gamma}(P_{R_1,f})$, $\ti_{\gamma}(P_{R_2,f})$  as $\ti_{R_1, \gamma}, \ti_{R_2, \gamma}$. Similarly, let $\ati_{R_1, \gamma-\epsilon}$ and $\ati_{R_2, \gamma-\epsilon}$  denote the approximation threshold implementation $\ati^{\epsilon, \delta}_{R_1, \gamma-\epsilon}$ and $\ati^{\epsilon, \delta}_{R_2, \gamma-\epsilon}$ respectively, for some inverse polynomial  $\epsilon$ and inverse subexponential function $\delta$ (in other words, $\log (1/\delta)$ is polynomial in $\lambda$). 

In this particular construction, $\cA$'s behavior can be described as follows: $\cA$ ``splits" the copy-protection state $\rho$ into two potentially entangled states $\sigma[R_1], \sigma[R_2]$. $\cA$ prepares $(\sigma[R_1], U_{R_1})$ with oracle access to $(\cO_1, \cO_2)$ as pirate program $\P_1$; and prepares $(\sigma[R_2], U_{R_2})$ with oracle access
$(\cO_1, \cO_2)$ as pirate program $\P_2$. Therefore, $\ti_{R_b, \gamma}$ and $\ati_{R_b, \gamma-\epsilon}$ both make oracle queries to $\cO_1, \cO_2$.

We can assume the joint state of $R_1, R_2$ has been purified and the overall state is a pure state over register $R_1, R_2, R_3$ where $P_1$ has only access to $R_1$ and $P_2$ has only access to $R_2$. 


\paragraph{Quantum Query Weight}
Let $\sigma$ be any quantum state of $R_1, R_2, R_3$. 
We consider the program $\P_1$. $\P_1$ has access to register $R_1$ and oracle access to $\cO = (\cO_1, \cO_2)$. We denote $\ket{\phi_i}$ to be the overall state of registers $R_1, R_2, R_3$ before $P_1$ makes $i$-th query to $\cO_1$, \emph{when it applies $\ati_{R_1,\gamma-\epsilon}$ on $\sigma[R_1]$}.
$$
\ket{\phi_i} = \sum_{x, v, z} \alpha_{x, v, z} \ket{x, v, z}.
$$
where $(x,v)$ is the query to oracle $\cO_1$ and $z$ is working space of $P_1$, the registers of $R_2, R_3$. 
Note that when $\ati_{R_1,\gamma-\epsilon}$ is applied on $\sigma[R_1]$, it in fact applies some unitary and eventually makes a measurement, during which the unitary makes queries to oracles $\cO_1, \cO_2$. Therefore such a query weight is well-defined.


We denote by $W_{1,A,i}$ to be the sum of squared amplitudes in $\ket{\phi_i}$, 
which are querying $\cO_1$ on input $(x,v)$ such that $v \in A \setminus \{0\}$:
$$
W_{1,A,i} = \sum_{x, v, z: v \in A\setminus \{0\}} \left|\alpha_{x, v, z}\right|^2
$$

Then we sum up all the squared amplitudes $W_{1,A,i}$ in all the queries made by $\P_1$ to $\cO_1$, where $v \in A \setminus \{0\}$. We denote this sum as $W_{1,A} = \sum_{i \in [\ell_1]} W_{1,A,i}$, 
where $\ell_1 = \ell_1(\lambda)$ is the number of queries made by $\P_1$ to $\cO_1$. 

Similarly, we write 
$W_{1,A^\perp} = \sum_{i \in [\ell_2]} W_{1,A^\perp,i} =  \sum_{i \in [\ell_2]} \sum_{x, v, z: v \in A^\perp \setminus\{0\}} \left|\alpha_{x, v, z}\right|^2$
to be the sum of squared amplitudes in $\ket{\phi_i}$ where $v \in A^\perp\setminus\{0\}$, in the $\ell_2$ queries made by $\P_1$ to $\cO_2$. 

Accordingly for the other program $\P_2$ and threshold implementation $\ati_{R_2,\gamma-\epsilon}$, we denote these sums of squared amplitudes as $W_{2,A} = \sum_{i \in [m_1]} W_{2,A,i}$ and $W_{2,A^\perp} = \sum_{i \in [m_2]} W_{2,A^\perp,i}$, where $m_1, m_2$ are the number of queries made by $\P_2$ to oracles $\cO_1, \cO_2$ respectively.

\vspace{1em}


\paragraph{Case One.} Fixing a function $f$, let $(\sigma, U_{R_1}, U_{R_2})$ be the two programs output by the adversary which are both tested $\gamma$-good respect to $f, D_f$ with some non-negligible probability.

Let $\cO_{\bot}$ be an oracle that always outputs $\bot$.  We hope one of the following will happen: 
\begin{enumerate}
    \item The program $(\sigma[R_1], U_{R_1})$ with oracle access to $\cO_1, \cO_{\bot}$ is tested $(\gamma - 2\epsilon)$-good respect to $f, D_f$, with non-negligible probability. 
    \item The program $(\sigma[R_2], U_{R_2})$ with oracle access to $\cO_{\bot}, \cO_2$ is tested $(\gamma - 2 \epsilon)$-good respect to $f, D_f$, with non-negligible probability. 
\end{enumerate}

Let $\widetilde{\ati}_{R_1, \gamma-\epsilon}$ be the same as $\ati_{R_1, \gamma-\epsilon}$ except with oracle access to $\cO_1, \cO_{\bot}$ and  $\widetilde{\ati}_{R_2, \gamma-\epsilon}$ be the same as $\ati_{R_2, \gamma-\epsilon}$ except with oracle access to $\cO_{\bot}, \cO_2$.  Similarly, let $\widetilde{\ti}_{R_b, \gamma - 2 \epsilon}$ be the same threshold implementation as ${\ti}_{R_b, \gamma - 2 \epsilon}$ except with oracle access to $\cO_1, \cO_{\bot}$ and $\cO_{\bot}, \cO_2$ respectively. 

Since  $(\sigma, U_{R_1}, U_{R_2})$ are both $\gamma$-good respect to $f, \cD_f$ with non-negligible probability, for some non-negligible function $\beta(\cdot)$, 
\begin{align*}
    \Tr[(\ti_{R_1, \gamma} \otimes \ti_{R_2, \gamma}) \cdot \sigma] \geq \beta(\lambda)
\end{align*}
From the property of the approximated threshold implementation (\cref{lem:ati_2d}), 
\begin{align*}
    \Tr[(\ati_{R_1, \gamma - \epsilon} \otimes \ati_{R_2, \gamma - \epsilon}) \cdot \sigma] \geq \beta(\lambda) - 2 \delta
\end{align*}
Thus, for any $b \in \{1, 2\}$, we have $\Tr[\ati_{R_b, \gamma - \epsilon} \cdot \sigma[R_b]] \geq \beta(\lambda) - 2 \delta$. Since $\delta$ is negligible, both probabilities are still non-negligible. 

Let $\eventone$ be the event denotes $\Tr[\widetilde{\ati}_{R_1, \gamma - \epsilon} \cdot \sigma[R_1]]$ is  non-negligible. If $\eventone$ happens, by \cref{cor:ati_thresimp}, 
\begin{align*}
    \Tr[\widetilde{\ti}_{R_1, \gamma - 2 \epsilon} \cdot \sigma[R_1]] \geq \Tr[\widetilde{\ati}_{R_1, \gamma - \epsilon} \cdot \sigma[R_1]] - \delta
\end{align*}
which is still non-negligible. In other words, $(\sigma[R_1], U_{R_1})$ with oracle access to $\cO_1, \cO_{\bot}$ is tested $(\gamma - 2\epsilon)$-good respect to $f, D_f$ with non-negligible probability.  
Similarly, define $\eventtwo$ as the program $(\sigma[R_2], U_{R_2})$ with oracle access to $\cO_{\bot}, \cO_2$ is $(\gamma - 2 \epsilon)$-good respect to $f, D_f$ with non-negligible probability.

\paragraph{Case Two.} Fixing a function $f$, let $(\sigma, U_{R_1}, U_{R_2})$ be the two programs output by the adversary which are both $\gamma$-good respect to $f, D_f$, with non-negligible probability.

If $\eventone \vee \eventtwo$ does not happen, we are in the case $\bareventone \wedge \bareventtwo$. By definition, there exist negligible functions $\negl_1, \negl_2$ such that 
\begin{align*}
    \Tr[\widetilde{\ati}_{R_1, \gamma - \epsilon} \cdot \sigma[R_1]]  \leq \negl_1(\lambda)   \quad\quad\quad
    \Tr[\widetilde{\ati}_{R_2, \gamma - \epsilon} \cdot \sigma[R_2]] &\leq \negl_2(\lambda) 
\end{align*}
We look at the following thought experiments: 
\begin{enumerate}
    \item We apply $\ati_{R_1, \gamma - \epsilon} \otimes \ati_{R_2, \gamma - \epsilon}$ on $\sigma$, by \cref{lem:ati_2d}, there exists a non-negligible function $\beta(\cdot)$ such that 
    \begin{align*}
        \Tr\left[(\ati_{R_1, \gamma -\epsilon} \otimes \ati_{R_2, \gamma - \epsilon}) \cdot \sigma\right] \geq \beta(\lambda) - 2 \delta.
    \end{align*}
    
    \item We apply $\ati_{R_1, \gamma - \epsilon} \otimes \widetilde{\ati}_{R_2, \gamma - \epsilon}$ on $\sigma$. We have, 
    \begin{align*}
        \Tr\left[(\ati_{R_1, \gamma -\epsilon} \otimes \widetilde{\ati}_{R_2, \gamma - \epsilon}) \cdot \sigma\right] \leq \Tr\left[(I \otimes \widetilde{\ati}_{R_2, \gamma - \epsilon}) \cdot \sigma\right] \leq \negl_2(\lambda). 
    \end{align*}
    \item Note that in 1 and 2, the only difference is the oracle access: in 1, it has oracle access to $\cO_1, \cO_2$; in 2, it has oracle access to $\cO_\bot, \cO_2$.
    Let $\sigma'$ be the state which we apply $(\ati_{R_1, \gamma -\epsilon} \otimes I)$ on $\sigma$ and obtain a outcome $0$, which happens with non-negligible probability.
    Let $W_{2, A}$ be the query weight defined on the state $\sigma'$.  We know that $W_{2, A}$ can not be negligible otherwise by \cref{thm:bbbv97_oraclechange} (BBBV), the probability difference in 1 and 2 can not be non-neglibile. 
    
    Define $M_{R_2}$ be the operator that measures a random query of ${\ati}_{R_2, \gamma-\epsilon}$ to $\cO_1$ and the query $(x, v)$ satisfies $v \in A \setminus \{0\}$. By the above discussion, there exists a non-negligible function $\beta_1(\cdot)$, 
    \begin{align*}
        \Tr\left[(\ati_{R_1, \gamma -\epsilon} \otimes M_{R_2}) \cdot \sigma\right] \geq \beta_1(\lambda).
    \end{align*}
    
    \item  We apply $\widetilde{\ati}_{R_1, \gamma - \epsilon} \otimes M_{R_2}$ on $\sigma$. We have, \begin{align*}
        \Tr\left[(\widetilde{\ati}_{R_1, \gamma - \epsilon} \otimes M_{R_2}) \cdot \sigma\right] \leq \Tr\left[(\widetilde{\ati}_{R_1, \gamma - \epsilon} \otimes I) \cdot \sigma\right] \leq \negl_1(\lambda). 
    \end{align*}
    
    \item By a similar argument of $3$, let $M_{R_1}$ be the operator that measures a random query of ${\ati}_{R_1, \gamma-\epsilon}$ to $\cO_2$ and the query $(x, v)$ satisfies $v \in A^\perp \setminus \{0\}$. There exists a non-negligible function $\beta_2(\cdot)$, 
    \begin{align*}
        \Tr\left[(M_{R_1} \otimes M_{R_2}) \cdot \sigma\right] \geq \beta_2(\lambda).
    \end{align*}
\end{enumerate}
Thus, in the case, one can extract a pair of vectors $(u, v) \in (A \setminus \{0\}) \times (A^\perp \setminus \{0\})$ with non-negligible probability. To conclude it, we have the following lemma, 
\begin{lemma} \label{lem:extraction_vectors}
Fixing a function $f$, let $(\sigma, U_{R_1}, U_{R_2})$ be the two programs output by the adversary which are both $\gamma$-good respect to $f, D_f$, with non-negligible probability. If $\eventone \vee \eventtwo$ does not happen, by randomly picking and measuring a query of $\ati_{R_1, \gamma-\epsilon}$ to $\cO_2$ and  a query of $\ati_{R_2, \gamma-\epsilon}$ to $\cO_1$, one can obtain a pair of vectors $(u, v) \in (A \setminus \{0\}) \times (A^\perp \setminus \{0\})$ with non-negligible probability. 
\end{lemma}

\vspace{1em}
By averaging over all randomness, we have the following lemma: 
\begin{lemma} \label{lem:eventone}
Let $\Pr[\eventone]$ be the probability of $\eventone$ taken over all randomness of $\AG^\As_{\cF, \cD, \gamma}(1^\lambda)$. If $\Pr[\eventone]$ is non-negligible, there exists an adversary $\As_1$ that wins $\LG^{\As_1}_{\cF, \cD, \gamma-2\epsilon}(1^\lambda)$ with non-negligible probability. 
\end{lemma}
\begin{proof}
The challenger in the copy protection security game plays as the quantum unlearnability adversary $\cA_1$ for function $f \gets \cF$, given only black-box access to $f$; we denote this black box as oracle $\cO_f$, which on query $\ket{x, z}$, answers the query with $\ket{x, f(x)+z}$.

Next, we show that $\cA_1$ can simulate the copy protection security game for $\cA$ using the information given and uses $\cA$ to quantumly learn $f$.
$\cA_1$ samples random $\lambda/2$-dimensional subspace $A$ over $\F$ and prepares the membership oracles (two unitaries) $U_A, U_A^\perp$ as well as state $\ket{A}$.

Using $U_A, U_A^\perp$ and given oracle access to $f$ in the unlearnability game, $\cA_1$ simulates the copy protection oracles $\cO_1, \cO_2$ for $\cA$ in the query phase of anti-piracy game.

There's one subtlety in the proof: $\cA_1$ needs to simulate the oracles in the anti-piracy game slightly differently: $\cA_1$ simulates the oracles with their functionalities partially swapped:
 \begin{align*}
    \cO_1'(x,v)&=\begin{cases}
     g(x)&\text{if }v\in A \text{ and }v\ne 0,\\
    \bot &\text{otherwise}.
    \end{cases} \\
    \cO_2'(x,v)&=\begin{cases}
    f(x)\oplus g(x)&\text{if }v\in A^\perp\ \text{ and }v\ne 0,\\
    \bot &\text{otherwise}.
     \end{cases}
    \end{align*}
    That is, a random function $g(x)$ is output when queried on $u \in A \setminus \{0\}$, and $f(x)\oplus  g(x)$ is output  when queried on $u \in A^\perp \setminus \{0\}$. The distributions of $\cO_1, \cO_2$ and $\cO'_1, \cO'_2$ are identical. Note that $g(x)$ can be simulated by a quantum secure PRF or a $2 t$-wise independent hash function where $t$ is the number of oracle queries made by $\As$ \cite{zhandry2012construct}. 
 
In the output phase, $\cA$ outputs $(\sigma, U_{R_1},U_{R_2})$ and sends to $\cA_1$.  $\As_1$ simply outputs $(\sigma[R_1], U_{R_1})$ with oracle access to $\cO'_1, \cO_{\bot}$. The program does not need access to oracle $f$ because $\cO'_1$ is only about $g(\cdot)$ and $\cO_{\bot}$ is a dummy oracle. 
If $\eventone$ happens, the program is a $(\gamma - 2\epsilon)$-good with non-negligible probability, by the definition of $\eventone$. Because $\Pr[E_1]$ is also non-negligible, $\As_1$ breaks $(\gamma-2 \epsilon)$-quantum-unlearnability of $\cF, \cD$.  \qed
\end{proof}

\begin{lemma} \label{lem:eventtwo}
Let $\Pr[\eventtwo]$ be the probability of $\eventtwo$ taken over all randomness of $\AG^\As_{\cF, \cD, \gamma}(1^\lambda)$. If $\Pr[\eventtwo]$ is non-negligible, there exists an adversary $\As_2$ that wins $\LG^{\As_2}_{\cF, \cD, \gamma-2\epsilon}(1^\lambda)$ with non-negligible probability. 
\end{lemma}
\begin{proof}[Proof Sketch]
The proof is almost identical to the proof for \cref{lem:eventtwo} except oracles  $\cO_1, \cO_2$ are simulated in the same way as that in the construction. $\cO_1(x,v)$ outputs $f(x) \oplus g(x)$ if $v \in A \setminus \{0\}$, and otherwise outputs $\bot$. Similarly, $\cO_2(x,v)$ outputs $g(x)$ if $v \in A^{\perp} \setminus \{0\}$, and otherwise outputs $\bot$ \qed
\end{proof}

As discussed above, if $\Pr[\eventone \vee \eventtwo]$ is non-negligible, we can break the quantum unlearnability. Otherwise, $\Pr[\bareventone \wedge \bareventtwo]$ is overwhelming. We show that in the case, one can use the adversary $\As$ to breaks the direct-product problem \cref{thm:restate_bound_twopts}. 

\begin{lemma} \label{lem:eventno}
Let $\Pr[\bareventone \wedge \bareventtwo]$ be the probability taken over all randomness of $\AG^\As_{\cF, \cD, \gamma}(1^\lambda)$. If $\Pr[\bareventone \wedge \bareventtwo]$ is non-negligible, there exists an adversary $\As_3$ that breaks  the direct-product problem. 
\end{lemma}
\begin{proof}
     The challenger in the copy protection security game plays as the adversary in breaking direct-product problem, denoted as $\cA_3$. In the reduction, $\cA_3$ is given the access to  membership oracles $U_A, U_A^\perp$ and one copy of $\ket{A}$.

Next, we show that $\cA_3$ can simulate the anti-piracy security game for $\cA$ using the information given and uses $\cA$ to obtain the two vectors.
$\cA_3$ samples $f \gets \cF$, and simulates a $\gamma$-anti-piracy game, specifically simulating the copy protection oracle $\cO_1, \cO_2$ for adversary $\cA$. 
In the output phase, $\cA$ outputs $(\sigma, U_{R_1}, U_{R_2})$. 

$\As_1$ upon taking the output, it randomly picks and measures a query of $\ati_{R_1, \gamma-\epsilon}$ to $\cO_2$ and  a query of $\ati_{R_2, \gamma-\epsilon}$ to $\cO_1$, and obtain a pair of vectors $(u, v)$. 
If $\bareventone \wedge \bareventtwo$ happens. By \cref{lem:extraction_vectors}, $(u, v)$ breaks the direct-product problem with non-negligible probability. 
Since $\Pr[\bareventone \wedge \bareventtwo]$ is non-negligible, the overall probability is non-negligible. 
\qed
\end{proof}

Note that the proof does not naturally extend to $q$-collusion resistant anti-piracy. We leave this as an open problem. 

\paragraph{The General Case.}  The above proof works for the general case, by simply doing the followings: 1. $\ti, \ati$ are now defined as the (approximated) projective measurement corresponding to the predicate $E_\lambda$; 2. In \cref{lem:eventone}, \ref{lem:eventtwo} and \ref{lem:eventno}, the randomness is taken over the general unlearnability game and copy-protection game.

\section{ Quantum Copy-Detection}

\subsection{Construction}
Now we construct a copy detection scheme for $\samp, \mathscr{F}, \mathscr{E}$. Let $\qm$ and $\wm$ be a public key quantum money scheme and a publicly extractable watermarking scheme for $\samp, \mathscr{F}, \mathscr{E}$, whose serial number space $\mathcal{S}_\lambda$ of $\qm$ is a subset of the message space $\mathcal{M}_\lambda$ of $\wm$. We construct a copy detection scheme in \cref{fig:cd_mini_scheme}.

\begin{figure}
    \centering
    \begin{gamespec}
    \begin{description}
    \item[] $\setup(1^\lambda)$: it runs $\wm.\setup(1^\lambda)$ to get $\xk, \mk$, let $\sk = \mk$ and $\pk = \xk$. 
    
    \item $\generate(\sk, f)$: 
    \begin{itemize}
        
        \item it runs $\qm.\gen(1^\lambda)$ to get a money state $|\$\rangle$ and a serial number $s$ (by applying \qm.\ver{} to the banknote); 
        
        \item let $\widetilde{f} = \wm.\mar(\mk, f, s)$ which is classical; \item it outputs the quantum state $\rho_f = (\widetilde{f}, |\$\rangle)$, and $\{U_{f, x}\}_{x \in [N]}$; 
        \item let  $\{U_{f, x}\}_{x \in [N]}$ describe the following unitary: on input a quantum state $\rho$, treat the first register as a classical function $g$, compute $g(x)$ in superposition. 
    \end{itemize}
    
    \item $\chk(\pk, \aux_f, (\rho_f, \{U_{f, x}\}_{x \in [N]}))$: 
    \begin{itemize}
        \item it parses and measures the first register, which is $(f', |\$'\rangle)$;
        \item it checks if $\qm.\ver(|\$'\rangle)$ is valid and it gets the serial number $s'$;
        \item it then checks if $s' = \wm.\extract(\pk = \xk, \aux_f, f')$;
        \item if all the checks pass, it outputs $0$; otherwise, it outputs $1$.
    \end{itemize}
    
    \end{description}
    \end{gamespec}
    \caption{Quantum copy detection scheme.}
    \label{fig:cd_mini_scheme}
\end{figure}

\subsection{Efficiency and Correctness}

First, for all $\lambda \in \N$, all efficient $\As$, every $f \gets (\As \Longleftrightarrow \samp(1^\lambda))$, the program output is $(\rho_f, \{U_{f, x}\}_{x \in [N]})$, we have
    $\compute(\rho_f, \{U_{f, x}\}_{x \in [N]}, x) = \tilde{f}(x)$, 
where $\tilde{f} = \wm.\mar(\mk, f, s)$ for some serial number $s$. From the correctness of $\wm$, it satisfies \textbf{unique output} and \textbf{functionality preserving} (with respect to $\mathscr{F}$). 

The correctness of $\chk$ comes from the correctness of $\wm.\extract$ and \textbf{unique serial number} property of $\qm$.
$\chk$ is a projection since $\qm.\ver$ is also a projection. Efficiency is straightforward.

\subsection{Security}
\begin{theorem} \label{thm:watermarktocd}
Assume $\qm$ is a quantum money scheme and $\wm$ is a $q$-collusion resistant for $\samp, \mathscr{E}$ with $\gamma$-unremovability, 
the above copy-detection scheme for $\samp, \mathscr{F}, \mathscr{E}$ has $q$-collusion resistant $\gamma$-copy-detection-security. 
\end{theorem}
\begin{proof}
    We prove the case for $q=1$. Let $\As$ be a QPT algorithm that tries to break the security of the copy detection scheme. 
    Let $(\sigma, U_{R_1}, U_{R_2})$ be the program output by $\As$ which wins the game $\DG^{\As}_{\samp, \mathscr{E}, \gamma}$. 
    
    To win the game, the program $(\sigma, U_{R_1}, U_{R_2})$ should pass the following two tests: 
    \begin{enumerate}
        \item Apply the projective measurement  (defined by $\chk(\pk, \aux_f, \cdot)$) on both $\sigma[R_1]$ and $\sigma[R_2]$, and both outcomes are $0$. 
        
        \item Let $\sigma'$ be the state that passes step 1. Then both programs $(\sigma'[R_1], U_{R_1}), \allowbreak (\sigma'[R_2], U_{R_2})$ are tested to be $\gamma$-good with non-negligible probability. 
    \end{enumerate}
    
    In our construction, $\chk$ first measures the program registers. The resulting state is $\tilde{f}_1, \tilde{f}_2, \sigma$, where $\tilde{f}_1, \tilde{f}_2$ are supposed to be classical (marked) circuits that computes $f$ and $\sigma$ are (possibly entangled) states that are supposed to be quantum money for each of the program. 
    
    Next, $\chk$ applies $\qm.\ver$ on both registers of $\sigma$ and computes serial numbers. Define $S_b$ be the random variable of $\qm.\ver$ applying on $\sigma[R_b]$ representing the serial number of $\rho_b$. Define $S$ be the random variable of $\qm.\ver(|\$\rangle)$ representing the serial number of the quantum money state in the $\generate$ procedure.

    Define $E$ be the event that  both $\wm.\extract(\xk, \aux_f, \tilde{f}_b) = S_b$ and at least one of $S_1, S_2$ is not equal to $S$.
    Define $E'$ be the event that both $S_1, S_2$ are equal to $S$ and both $\wm.\extract(\xk, \aux_f, \tilde{f}_b) = S_b$.
    If $\tilde{f}_1, \tilde{f}_2, \sigma$ passes the step 1, exactly one of $E$ and $E'$ happens.
    
    In step 2, it simply tests if $\tilde{f}_1$ and $\tilde{f}_2$ are $\gamma$-good with respect to $f, E_\lambda$. Since $\tilde{f}_1, \tilde{f}_2$ are classical circuits, it is equivalent to check whether they work correctly on at least $\gamma$ fraction of all inputs. If it passes step 2, we have for all $b \in \{1, 2\}$,  
       $\Pr_r[E_\lambda(\tilde{f}_b, f, r) = 0] \geq \gamma$. 
    
    Therefore, the probability of $\As$ breaks the security game is indeed, 
    \begin{align*}
        &\Pr_{(\tilde{f}_1, \tilde{f}_2, \sigma)}\left[\forall b, \Pr_r[E_\lambda(\tilde{f}_b, f, r) = 0] \geq \gamma \right] \\
        =& \Pr_{(\tilde{f}_1, \tilde{f}_2, \sigma)}\left[(E \vee E') \wedge \forall b, \Pr_r[E_\lambda(\tilde{f}_b, f, r) = 0] \geq \gamma\right] \\
        \leq&  \Pr_{(\tilde{f}_1, \tilde{f}_2, \sigma)}\left[E \wedge \forall b, \Pr_r[E_\lambda(\tilde{f}_b, f, r) = 0] \geq \gamma\right] + \Pr_{(\tilde{f}_1, \tilde{f}_2, \sigma)}[E']
    \end{align*}
    Note that the probability is taken over the randomness of $\DG^{\As}_{\samp, \mathscr{E}, \gamma}$. Next we are going to show both probabilities are negligible, otherwise we can break the quantum money scheme or watermarking scheme.
    
    \begin{claim}
        $\Pr_{(\tilde{f}_1, \tilde{f}_2, \sigma)}[E'] \leq \negl(\lambda)$. 
    \end{claim}
    \begin{proof}
         It corresponds to the security game of the quantum money scheme. Assume $\Pr[E']$ is non-negligible, we can construct an adversary $\Bs$ for the quantum money scheme with non-negligible advantage. 
        Given a quantum money state $|\$\rangle$, the algorithm $\Bs$ does the following (it simulates the challenger for the copy-detection scheme): 
        \begin{itemize}
            \item It first runs $\wm.\setup(1^\lambda)$ to get $\xk, \mk$ and let $\sk = \mk$ and $\pk = \xk$. 
            \item It interacts with $\As$ and samples $f$. 
            \item Instead of sampling a new quantum money state, it uses the state $|\$\rangle$. Let $s = \ver(|\$\rangle)$ and $\widetilde{f} \gets \wm.\mar(\mk, f, s)$. It gives the instance $\rho_f = (\widetilde{f}, |\$\rangle)$. 
            \item When $\As$ outputs $(\tilde{f}_1, \tilde{f}_2, \sigma)$, $\Bs$ outputs $\sigma$. 
        \end{itemize}
        Thus $\Pr[E']$ is exact the probability that both verification gives $s$. \qed        
    \end{proof}
    
    \begin{claim}
        $\Pr_{(\tilde{f}_1, \tilde{f}_2, \sigma)}\left[E \wedge \forall b, \Pr_r[E_\lambda(\tilde{f}_b, f, r) = 0]\right] \leq \negl(\lambda)$. 
    \end{claim}
    \begin{proof}
        It corresponds to the security game of the underlying watermarking scheme. 
        Since if $E$ happens, at least one of the circuit has different mark than $s$ and it satisfies the correctness test $F$. 
        The reduction is the following ($\Bs$ simulates the challenger for the copy-detection scheme): 
        \begin{itemize}
            \item Given $\xk, \aux_f$ in the watermarking security game, $\Bs$ prepares a quantum money state $|\$\rangle$ with serial number $s$ and gets the marked circuit $\widetilde{f}$ whose marking is $s$.
            \item It prepares $\rho_f = (\widetilde{f}, |\$\rangle)$  and  feeds it to $\As$.
            \item When $\As$ outputs outputs $(\tilde{f}_1, \tilde{f}_2, \sigma)$, $\Bs$ outputs $\tilde{f}_b$ whose mark is not $s$, i.e, $\extract(\xk, \aux_f,  \tilde{f}_b) \ne s$. 
        \end{itemize} 
        When $\As$ succeeds, $\Bs$ breaks the security of the watermarking scheme.  \qed
    \end{proof}

    Thus, the probability of $\As$ breaks the game is negligible. \qed
\end{proof}

It is natural to extend the proof to $q$-collusion resistance. We put the proof sketch in \cref{sec:cd_q_sketch}.

Combining with the watermarking primitives (see examples in Appendix~\ref{sec:wmapp}), we can get the corresponding copy-detection schemes.

\bibliographystyle{alpha}
\bibliography{main.bib}

\appendix

\section{Basics of Quantum Computation and Quantum Information}\label{sec:basic_qc}
For completeness, we provide some of the basic definitions of quantum computing and quantum information, for more details see \cite{nc02}.

\paragraph{Quantum states}
Let $\mathscr{H}$ be a finite Hilbert space. Quantum states over $\mathscr{H}$ are positive semi-definite operators from $\mathscr{H}$ to $\mathscr{H}$ with unit trace. These are called density matrices, denoted by $\rho$ or $\sigma$ in this paper. 

A quantum state over $\mathscr{H}=\C^2$ is called qubit, which can be represented by the linear combination of the standard basis $\{\ket{0}, \ket{1}\}$. More generally, a quantum system over $(\C^2)^{\otimes n}$ is called an $n$-qubit quantum system for $n\in \N_+$. 

A pure state can be represented by a unit vector in $\C^n$. The standard basis of the Hilbert space of $n$-qubit pure states is denoted by $\{\ket{x}\}$, where $x\in \{0,1\}^n$. If a state $\ket{\phi}$ is a linear combination of several $\ket{x}$, we say it is in ``superposition''.

A mixed state is a collection of pure states $\ket{\phi_i}$ for $i\in [n]$, each with associated probability $p_i$, with the condition $p_i\in [0,1]$ and $\sum_{i=1}^n p_i = 1$. A mixed state can also be represented by the density matrix: $\rho:= \sum_{i=1}^n p_i \ket{\phi_i}\bra{\phi_i}$. 

\emph{Partial Trace}. 
For a quantum state $\sigma$ over two registers $R_1, R_2$ (i.e. Hilbert spaces $\mathscr{H}_{R_1}, \mathscr{H}_{R_2}$), we denote the state in $R_1$ as $\sigma[R_1]$, where $\sigma[R_1]= \Tr_2[\sigma]$ is a partial trace of $\sigma$. Similarly, we denote $\sigma[R_2]= \Tr_1[\sigma]$.

\emph{Purification of mixed states}. For a mixed state $\rho$ over $\mathscr{H}_A$, there exists another space $\mathscr{H}_B$ and a pure state $\ket \psi$ over $\mathscr{H}_A \otimes \mathscr{H}_B$ such that $\rho$ is a partial trace of $\ket \psi \bra \psi$ with respect to $\mathscr{H}_B$. 



\begin{definition}[Trace distance]
Let $\rho,\sigma\in \C^{2^n\times 2^n}$ be the density matrices of two quantum states. The trace distance between $\rho$ and $\sigma$ is
\begin{align*}
    \|\rho-\sigma\|_{\mathrm{tr}}:=\frac{1}{2}\sqrt{\Tr[(\rho-\sigma)^\dagger (\rho-\sigma)]},
\end{align*}
\end{definition}

\paragraph{Quantum Measurements}
In this work, we will use the following general form of measurements.
\begin{definition}[Positive operator-valued measure, POVM]
A positive operator-valued measure (POVM) $\mathcal{M}$ is specified by a finite index set $\mathcal{I}$ and a set $\{M_i\}_{i\in \cal I}$ of Hermitian positive semi-definite matrices $M_i$ such that $\sum_{i\in \cal I}M_i = \mathbf{I}$. 

When applying $\mathcal{M}$ to a quantum state $\rho$, the outcome is $i$ with probability $p_i = \Tr[\rho P_i]$ for all $i\in \cal I$.
\end{definition}

To characterize the post-measurement states, we define the quantum measurements as follows.
\begin{definition}[Quantum measurement]
A quantum measurement $\cal E$ is specified by a finite index set $\cal I$ and a set $\{E_i\}_{i\in \cal I}$ of measurement operators $E_i$ such that $\sum_{i\in \cal I}E_i^\dagger E_i=\mathbf{I}$.

When applying $\cal E$ to a quantum state $\rho$, the outcome is $i$ with probability $p_i = \Tr[\rho E_i^\dagger E_i]$ for all $i\in \cal I$. Furthermore, conditioned on the outcome being $i$, the post-measurement state is $E_i\rho E_i^\dagger / p_i$.
\end{definition}

Note that POVM $\cal M$ and quantum measurement $\cal E$ are related by setting $M_i = E_i^\dagger E_i$. In this case, we say that $\cal E$ is an implementation of $\cal M$. The implementation of a POVM may not be unique.

\begin{definition}[Projective measurement and projective POVM]
A quantum measurement ${\cal E}$ is projective if for all $i\in \cal I$, $E_i$ is a projection, i.e., $E_i$ is Hermitian and $E_i^2 = E_i$.

Similarly, a POVM $\cal M$ is projective if each $M_i$ is projection for $i\in \cal I$.
\end{definition}

\section{Cryptographic Primitives}
\subsection{Public-key Quantum Money}
\begin{definition}[Public Key Quantum Money] 
\label{def:money}
A public-key (publicly-verifiable) quantum money consists of the following algorithms:

\begin{itemize}
    \item $\keygen(1^\lambda):$ takes as input a security parameter $\lambda$, and generates a key pair $(\sk, \pk)$.

\item $\gennote(\sk): $
takes   a secret key $\sk$ and generates a quantum banknote state $\ket{\$}$.

\item  $\ver(\pk, \ket{\$'}):$ takes a public key $\pk$,  and a claimed money state $\ket{\$'}$, and outputs either 1 for accepting or 0 for rejecting.
\end{itemize}

A secure public-key quantum money should satisfy the following properties:
\begin{description}
\item \textbf{Verification Correctness}: there exists a negligible function $\negl(\cdot)$ such that the following holds for any $\lambda \in \N$,
$$
\Pr_{(\sk, \pk) \gets \keygen(1^\lambda)}[\ver(\pk, \gennote(\sk)) = \text{ 1}] \geq 1 - \negl(\lambda)
$$

\item \textbf{Unclonable Security}: 
Suppose a QPT adversary is given $q = \poly(\lambda)$ number of valid banknotes $\{\rho_i\}_{i \in [q]}$ and then generates $q' = q+1$ banknotes $\{\rho'_j\}_{j \in [q']}$ where $\rho'_j$ are potentially entangled, there exists a negligible function $\negl(\cdot)$, for all $\lambda \in \N$, 
$$
\Pr_{(\sk, \pk) \gets \keygen(1^\lambda)} \left[\forall i \in [q'], \ver(\pk, \rho'_j) = 1: \{\rho'_j\} \gets \cA(1^\lambda, \{\rho_i\} \right] \leq \negl(\lambda)
$$

\end{description}
\end{definition}

\begin{remark}
In rest of the paper, $q$ is set to be 1 for simplicity, and the scheme satisfies unclonable security if $\cA$ cannot produce two banknotes that pass verification. \cite{aaronson2012quantum} shows that any public-key quantum money scheme that satisfies security when $q = 1$, can be generalized to a scheme that is secure when $q = \poly(\lambda)$, using quantum-secure digital signatures.
\end{remark}

A non-perturb property is also required. That is, one can verify a quantum banknote polynomially many times and the banknote is still a valid banknote. Since $\ver$ is almost a deterministic function, by Gentle Measurement Lemma (\cref{lem:gentle_measure}), the above definition implies the non-perturb property. 


In some settings, instead of outputting $0/1$, $\ver$ is required to output either $\bot$ which indicates the verification fails, or a serial number $s \in \mathcal{S}_\lambda$ if it passes the verification. In this case, the scheme should satisfy the following correctness (unique serial number property) and unclonable security \cite{zhandry2017quantum}: 

\begin{description}
\item \textbf{Unique Serial Number}:  For a money state $|\$\rangle$, let      $H_\infty(|\$\rangle) = - \log \min_s \Pr[\ver(|\$\rangle) = s]$. 
    We say a quantum scheme has unique serial number property, if $\mathbb{E}[H_\infty(|\$\rangle)]$ is negligible for all $\lambda$, $(\sk,\pk) \gets \keygen(1^\lambda)$ and $|\$\rangle$ is sampled from $\gennote(\sk)$. 
    
\item \textbf{Unclonable Security}:   
Consider the following game with a challenger and an adversary, 
    \begin{enumerate}
        \item The challenger runs $(\sk, \pk) \gets \keygen(1^\lambda)$ and $|\$\rangle \gets \gennote(\sk)$, it then runs $\ver$ to get a serial number $s$. 
        \item $\As$ is given  the public key $\pk$, the banknote $|\$\rangle$ and the serial number $s$. 
        \item $\As$ produces $\sigma^*$ (which contains two separate registers, but they may be entangled) and denotes $\sigma_1 = {\sf Tr}_2[\sigma^*]$ and $\sigma_2 = {\sf Tr}_1[\sigma^*]$. 
        \item $\As$ wins if and only if $\ver(\sigma_1) = \ver(\sigma_2) = s$. 
    \end{enumerate}
    We say a public key quantum money scheme is secure, if for all QPT $\As$, it wins the above game with negligible probability in $\lambda$.
\end{description}

\subsection{Obfuscation}
\begin{definition}[Virtual Black-Box Obfuscation, \cite{barak2001possibility}] 
An obfuscator $\cO$ (with auxiliary input) for a collection of circuits $\cC = \bigcup_{\lambda \in \N} \cC_\lambda$ 
is a (worst-case) VBB obfuscator if it satisfies:
\begin{itemize}
    \item \textbf{Functionality-Preserving:} 
    For every $C \in \cC$, every input $x$, $\Pr[\cO(C)(x) = C(x)] = 1$.
    
    \item \textbf{Virtual Black-Box:}
    For every poly-size adversary $\cA$, there exists a poly-size simulator
$\cS$, such that for every $\lambda \in \N$, auxiliary input $\aux \in \{0,1\}^{\poly(\lambda)}$, and every predicate $\pi : \cC_\lambda \to \{0, 1\}$, and every $\cC \in \cC_\lambda$:
$$
\left| \Pr_{\cA,\cO}[\cA(\cO(C), \aux) = \pi(C)] - \Pr_{\cS}[\cS^C(1^\lambda, \aux)) = \pi(C)] \right| \leq \negl(\lambda)
$$
where the probability is over $C \gets \cC_\lambda$, and the randomness of the algorithms $\cO,\cA$ and $\cS$.
\end{itemize}

\end{definition}

\section{Examples of Watermarking Primitives}
\label{sec:wmapp}

Let us look at how the definitions in \cite{cohen2018watermarking,goyal2019watermarking} fit into our frameworks.
\begin{enumerate}
    \item Watermarkable PRF in \cite{cohen2018watermarking}: 
        \begin{itemize}
            \item $\setup(1^\lambda) = (\wpp, \xk, \mk)$; 
            \item $\samp(1^\lambda, \wpp)$ samples a PRF key $k$, $f = \prf(k, \cdot)$, $\aux_f = \bot$. 

            \item $F_\lambda(\tilde{f}, f, r)$ is $0$ if and only if it samples a random input $x$ (according to $r$), and $\tilde{f}(x) = f(x)$. 
            
            \item Unremovability is defined by $E_\lambda = F_\lambda$. $\gamma  = 1/2 + 1/ \poly(\lambda)$. 
        \end{itemize}
    \item Watermarkable signature  in \cite{goyal2019watermarking}: 
            \begin{itemize}
                \item $\setup(1^\lambda) = (\wpp, \xk, \mk)$; 
                \item $\samp(1^\lambda, \wpp)$ samples a pair of keys $\vk, \sk$ and we interpret $f = \sign(\sk, \cdot) || \vk$, $\aux_f = \vk$. 
                \item $F_\lambda(\tilde{f}, f, r)$ is $0$ if and only if $\ver(\vk, m, \tilde{f}(m)) = 1$, where $\vk$ is decoded from $f$ and $m$ is sampled by $r$. 
                
                \item Unremovability:  $E_\lambda = F_\lambda$. 
                 $\gamma$ is inverse polynomial. 
            \end{itemize}
    
    \item Watermarkable public key encryption in \cite{goyal2019watermarking}: 
    \begin{itemize}
     \item  $\setup(1^\lambda) = (\wpp, \xk, \mk)$; 
     \item $\samp(1^\lambda, \wpp)$ is defined below:
        \begin{itemize}
            \item It samples $(\pk, \sk) \gets {\sf PKEGen}(1^\lambda, \wpp)$; 
            \item  $f = {\sf Dec}(\sk, \cdot) || \pk$, $
            \aux_f= \pk$. 
        \end{itemize}
     
     \item $F_\lambda(\tilde{f}, f, r)$ is defined as: 
        \begin{itemize}
            \item Decode $\pk$ from $f$, sample $m$ according to $r$; 
            \item Let $\ct = \enc(\pk, m)$; 
            \item It outputs $0$ if and only if  $\tilde{f}(\ct) = m$. 
        \end{itemize}

     \item Unremovability: $E_\lambda(\tilde{f}, f, r)$ defined as: 
     \begin{itemize}
         \item Decode $\pk$ from $f$, sample $b$ according to $r$; \item Decode $\aux = (m_0, m_1)$ from $\tilde{f}$; if $m_0 = m_1$, outputs $1$; 
         \item  Let $\ct = \enc(\pk, m_b)$;
         \item It outputs $0$ if and only if  $\tilde{f}(\ct) = m_b$. 
     \end{itemize}
     And, $\gamma = 1/2 + 1/\poly(\lambda)$. 
    \end{itemize}
\end{enumerate}

\section{Missing Details}
\subsection{Proof of \cref{cor:ati_thresimp}}\label{sec:proof_ati_th}
\begin{corollary}[\cref{cor:ati_thresimp}, restated]
    For any $\epsilon, \delta, \gamma, \cP, D$, the algorithm of measurement $\ati^{\epsilon, \delta}_{\cP, D, \gamma}$ that satisfies the followings: 
    \begin{itemize}
        \item For all quantum state $\rho$, $\Tr[\ati^{\epsilon, \delta}_{\cP, D, \gamma-\epsilon}\cdot \rho]\geq \Tr[\ti_\gamma(\cP_D)\cdot \rho]-\delta$.
        
        \item By symmetry, for all quantum state $\rho$, $\Tr[\ti_{\gamma-\epsilon}(\cP_D)\cdot \rho] \geq \Tr[\ati^{\epsilon, \delta}_{\cP, D, \gamma}\cdot \rho] -\delta$.
        
        \item 
        For all quantum state $\rho$, 
        let $\rho'$ be the collapsed state after applying  $\ati^{\epsilon, \delta}_{\cP, D, \gamma}$ on $\rho$. Then, $\Tr[\ti_{\gamma-2 \epsilon}(\P_D)\cdot \rho'] \geq 1 - 2\delta$. 
        \item
        The expected running time is the same as $\api_{\cP, D}^{\epsilon, \delta}$.
    \end{itemize}
\end{corollary}

We give the following fact before proving the corollary. 
\begin{fact}\label{fac:shift_dist_geq}
Let $D_0, D_1$ be two real-valued probability distributions with shift distance $\Delta_{\shift}^\epsilon=\delta$. Then, we have
\begin{align*}
    \Pr[D_0\geq x-\epsilon]\geq&~ \Pr[D_1 > x] - \delta, ~~~ \text{and}\\
    \Pr[D_1\geq x-\epsilon]\geq&~ \Pr[D_0 > x] - \delta
\end{align*}
\end{fact}
\begin{proof}
We prove the first inequality. By the definition of shift distance, we have
\begin{align*}
    \Pr[D_0 \leq x-\epsilon]\leq \Pr[D_1 \leq x] + \delta.
\end{align*}
Then, $\Pr[D_0 \geq x-\epsilon] = 1-\Pr[D_0 \leq x-\epsilon]  \geq 1- \Pr[D_1 \leq x] - \delta  = \Pr[D_1 \geq x] - \delta$.
The second inequality can be proved in a symmetric way. \qed
\end{proof}

Now, we prove the \cref{cor:ati_thresimp} in below.

\begin{proof}
By \cref{thm:api_projimp}, we know that there exists an algorithm $\api_{\cP, D}^{\epsilon, \delta}$ that approximates the measurement of $\projimp(\P_D)$, i.e., 
\begin{align*}
    \Delta^\epsilon_{\shift}(\api^{\epsilon, \delta}_{\cP, D}, \projimp(\cP_D)) \leq \delta.
\end{align*}
In particular, for any pure quantum state $\ket{\psi}$, let $D_A$ be the distribution of $\api_{\cP, D}^{\epsilon, \delta}(\ket{\psi})$ and $D_P$ be the distribution of $\projimp(\cP_D)(\ket{\psi})$. 

Then, by~\cref{fac:shift_dist_geq}, we have
\begin{align*}
    \Pr[D_A\geq \gamma-\epsilon]\geq \Pr[D_P\geq \gamma]-\delta,
\end{align*}

Hence, by the definition of threshold implementation (\cref{def:thres_implement}) and the construction of the algorithm $\ati_{\cP,D,\gamma}^{\epsilon, \delta}$, we can get
\begin{align*}
    \Tr\left[\ati_{\cP,D,\gamma-\epsilon}^{\epsilon, \delta}\ket{\psi}\bra{\psi}\right] \geq \Tr\Big[\ti_\gamma(\cP_D)\ket{\psi}\bra{\psi}\Big] - \delta.
\end{align*}

Note that mixed state is just a convex combination of pure states. Hence, by the linearity of trace, for any mixed state $\rho$, we have
\begin{align*}
    \Tr\Big[\ati^{\epsilon, \delta}_{\cP, D, \gamma-\epsilon}\cdot \rho\Big]\geq \Tr\Big[\ti_\gamma(\cP_D)\cdot \rho\Big] -\delta,
\end{align*}
which proves the first bullet. The second bullet follows the same idea by symmetry. 

For the third bullet, notice that the measurement algorithms $\ati_{\cP, D, \gamma}^{\epsilon, \delta}$ and $\api_{\cP, D}^{\epsilon, \delta}$ do the same thing to the quantum state. So, $\rho'$ is also the collapsed state after the measurement of $\api_{\cP, D}^{\epsilon, \delta}(\rho)$.

Since we assume that the outcome of $\ati_{\cP, D, \gamma}^{\epsilon, \delta}(\rho)$ is 0, it implies the corresponding outcome of $\api_{\cP, D}^{\epsilon, \delta}(\rho)$ is at least $\gamma$. 

By \cref{thm:api_projimp}, $\api_{\cP, D}^{\epsilon, \delta}$ is $(\epsilon, \delta)$-almost projective, which means that if we apply $\api_{\cP, D}^{\epsilon, \delta}$ (again) to $\rho'$, the outcome satisfies
\begin{align*}
    \Pr[\api_{\cP, D}^{\epsilon, \delta}(\rho')<\gamma-\epsilon]<\delta.
\end{align*}
\cref{thm:api_projimp} also provides that the shift distance between $\projimp$ and $\api$ is small, which means
\begin{align*}
    \Pr[\projimp(\cP_D)(\rho')\leq \gamma - 2\epsilon] \leq \Pr[\api_{\cP, D}^{\epsilon, \delta}(\rho')<\gamma-2\epsilon + \epsilon] + \delta \leq 2\delta.
\end{align*}
Hence,
\begin{align*}
    \Tr[\ti_{\gamma -2\epsilon}\cdot \rho']= 1- \Pr[\projimp(\cP_D)(\rho')\leq \gamma - 2\epsilon]\geq 1- 2\delta.
\end{align*}

The third bullet easily follows from the construction. \qed
\end{proof}

\subsection{Proof of \cref{lem:ati_2d}}\label{sec:proof_ati_2d}
\begin{lemma}[\cref{lem:ati_2d}, restated]
Let $\cP_1$ and $\cP_2$ be two  collections of projective measurements and $D_1$ and $D_2$ be any probability distributions defined on the index set of $\cP_1$ and $\cP_2$ respectively. 
For any $0<\epsilon, \delta, \gamma < 1$, the algorithms $\ati^{\epsilon, \delta}_{\cP_1, D_1, \gamma}$ and $\ati^{\epsilon, \delta}_{\cP_2, D_2, \gamma}$ satisfy the followings:
\begin{itemize}
    \item For any bipartite (possibly entangled, mixed) quantum state $\rho\in \mathscr{H}_{\cal L}\otimes \mathscr{H}_{\cal R}$,
    \begin{align*}
        \Tr\big[\big(\ati_{\cP_1, D_1, \gamma-\epsilon}^{\epsilon, \delta}\otimes \ati_{\cP_2, D_2, \gamma-\epsilon}^{\epsilon, \delta}\big)\rho\big] \geq \Tr\big[ \big(\ti_{\gamma}(\cP_{D_1})\otimes \ti_\gamma(\cP_{D_2})\big)\rho \big] - 2\delta.
    \end{align*}
    \item For any (possibly entangled, mixed) quantum state $\rho$, 
        let $\rho'$ be the collapsed state after applying  $\ati^{\epsilon, \delta}_{\cP_1, D_1, \gamma}\otimes \ati^{\epsilon, \delta}_{\cP_2, D_2, \gamma}$ on $\rho$  (and normalized). Then, 
        \begin{align*}
            \Tr\big[ \big(\ti_{\gamma-2\epsilon}(\cP_{D_1})\otimes \ti_{\gamma-2\epsilon}(\cP_{D_2})\big) \rho'\big]\geq  1- 4 \delta.
        \end{align*}
\end{itemize}
\end{lemma}
\begin{proof}
We use the hybrid argument to show that $\ati_{\cP_1, D_1, \gamma-\epsilon}^{\epsilon, \delta}\otimes \ati_{\cP_2, D_2, \gamma-\epsilon}^{\epsilon, \delta}$ approximates $\ti_{\gamma}(\cP_{D_1})\otimes \ti_\gamma(\cP_{D_2})$. 

For brevity, let $\ati_1$ denote $\ati_{\cP_1, D_1, \gamma-\epsilon}^{\epsilon, \delta}$ and $\ati_2$ denote $\ati_{\cP_2, D_2, \gamma-\epsilon}^{\epsilon, \delta}$. Similarly, let $\ti_1$ denote $\ti_\gamma(\cP_{D_1})$ and $\ti_2$ denote $\ti_\gamma(\cP_{D_2})$.
We first show that, 
\begin{align}\label{eq:thre_hy0_hy1}
    \Tr[ (\ti_1\otimes \ti_2)\rho] \leq \Tr[ (\ti_1\otimes \ati_2)\rho] + \delta.
\end{align}
Note that $\rho$ is a bipartite quantum state in $\mathscr{H}_{\cal L}\otimes \mathscr{H}_{\cal R}$. So, we can consider $\ti_1\otimes \ti_2$ as a measurement performed by two parties $\mathcal{L}$ and ${\cal R}$. In this way, we can write the trace as the probability that ${\cal L}$ gets outcome 0 and ${\cal R}$ gets outcome 0:
\begin{align*}
    \Tr[ (\ti_1\otimes \ti_2)\rho] = \Pr[{\cal L}\leftarrow 0 \wedge {\cal R}\leftarrow 0].
\end{align*}

We can see that from $\ti_1\otimes \ti_2$ to $\ti_1\otimes \ati_2$, ${\cal L}$ performs the same measurement. 
Hence, we can condition on the event that ${\cal L}$ gets outcome 0 and let $\rho_1$ be the remaining mixed state that traced out the ${\cal L}$-part. 
Then, we get that
\begin{align*}
    \Pr[{\cal L} \leftarrow 0 \wedge R\leftarrow 0] = &~ \Pr[{\cal L}\leftarrow 0] \cdot \Pr[R\leftarrow 0 | {\cal L}\leftarrow 0]\\
    = &~ \Pr[{\cal L}\leftarrow 0] \cdot \Tr[\ti_2 \cdot \rho_1]\\
    \leq &~ \Pr[{\cal L}\leftarrow 0] \cdot (\Tr[\ati_2\cdot \rho_1] + \delta)\\
    \leq &~ \Tr[(\ti_1\otimes \ati_2)\rho] + \delta,
\end{align*}
where the first inequality follows from \cref{cor:ati_thresimp} and the last step follows from  $\Pr[{\cal L}\leftarrow 0]\leq 1$.

The next step is to show that:
\begin{align}\label{eq:thre_hy1_hy2}
    \Tr[(\ti_1\otimes \ati_2)\rho]\leq \Tr[(\ati_1\otimes \ati_2)\rho] + \delta.
\end{align}

In this case, ${\cal R}$ performs the same measurement. We can condition on the event that ${\cal R}$ gets outcome 0 and let $\rho_2$ be the remaining mixed state traced out the ${\cal R}$-part.

Hence, by a similar argument, we get that
\begin{align*}
    \Tr[(\ti_1\otimes \ati_2)\rho] = &~ \Pr[{\cal L}\leftarrow 0\wedge {\cal R}\leftarrow 0]\\
    = &~ \Pr[{\cal R}\leftarrow 0]\cdot \Pr[{\cal L}\leftarrow 0|{\cal R}\leftarrow 0]\\
    = &~ \Pr[{\cal R}\leftarrow 0]\cdot \Tr[\ti_1\cdot \rho_2]\\
    \leq &~ \Pr[{\cal R}\leftarrow 0]\cdot (\Tr[\ati_1\cdot \rho_2] + \delta)\\
    \leq &~ \Tr[(\ati_1\otimes \ati_2)\rho] + \delta.
\end{align*}

Combining the \cref{eq:thre_hy0_hy1} and \cref{eq:thre_hy1_hy2} proves the first bullet of the lemma:
\begin{align*}
    \Tr[ (\ti_1\otimes \ti_2)\rho]\leq &~ \Tr[ (\ti_1\otimes \ati_2)\rho] + \delta\\
    \leq &~  \Tr[(\ati_1\otimes \ati_2)\rho] + 2\delta.
\end{align*}

For the second part of the lemma, the trace can also be written as 
\begin{align*}
    \Tr\big[ \big(\ti_{\gamma-2\epsilon}(\cP_{D_1})\otimes \ti_{\gamma-2\epsilon}(\cP_{D_2})\big) \rho'\big] = &~  \Pr[{\cal L}\leftarrow 0 \wedge {\cal R}\leftarrow 0]\\
    = &~ \Pr[{\cal L}\leftarrow 0] \cdot \Pr[{\cal R}\leftarrow 0| {\cal L}\leftarrow 0],
\end{align*}
where ${\cal L}$ and ${\cal R}$ are now performing measurements on $\rho'$.

We first rewrite the term $\Pr[{\cal L}\leftarrow 0]$ as
\begin{align*}
    \Pr[{\cal L}\leftarrow 0] =  \Tr[(\ti_{\gamma-2\epsilon} (\cP_{D_1})\otimes {\bf I})\rho'].
\end{align*}
We can see that this measure process is equivalent to the following process:
\begin{enumerate}
    \item ${\cal R}$ first performs the measurement $\ati^{\epsilon, \delta}_{\cP_2, D_2, \gamma}$ on the ${\cal R}$-part of $\rho$ and gets a state $\rho_1$ such that $\Tr_{\cal L}[\rho_1]=\Tr_{\cal L}[\rho']$.
    \item ${\cal L}$ measures $\ati^{\epsilon, \delta}_{\cP_1, D_1, \gamma}$ on $\Tr_{\cal R}[\rho_1]$ and get the collapsed state $\rho_2$ such that $\rho_2 = \Tr_{\cal R}[\rho']$.
    \item ${\cal L}$ measures $\ti_{\gamma-2\epsilon}(\cP_{D_1})$ on $\rho_2$.
\end{enumerate}

Hence, we have
\begin{align*}
    \Tr[(\ti_{\gamma-2\epsilon} (\cP_{D_1})\otimes {\bf I})\rho'] = \Tr[\ti_{\gamma-2\epsilon} (\cP_{D_1}) \cdot \rho_2],
\end{align*}

By \cref{cor:ati_thresimp} (the third bullet),
\begin{align*}
    \Tr[\ti_{\gamma-2\epsilon} (\cP_{D_1}) \cdot \rho_2] \geq 1-2\delta.
\end{align*}
Hence, we get that $\Pr[{\cal L}\leftarrow 0] \geq 1-2\delta$.

For the second term $\Pr[{\cal R}\leftarrow 0|{\cal L}\leftarrow 0]$, it can be written as
\begin{align*}
    \Pr[{\cal R}\leftarrow 0|{\cal L}\leftarrow 0] = \Tr[({\bf I}\otimes \ti_{\gamma -2\epsilon}(\cP_{D_2}))\cdot \rho_3] = \Tr[\ti_{\gamma -2\epsilon}(\cP_{D_2})\cdot \rho_4],
\end{align*}
where $\rho_3$ is the collapsed state conditioned on the outcome of ${\cal L}$ being 0 and $\rho_4 = \Tr_{\cal L}[\rho_3]$.

This measure process is equivalent to the followings:
\begin{enumerate}
    \item ${\cal L}$ first performs two consecutive measurements $\ati^{\epsilon, \delta}_{\cP_1, D_1, \gamma}$ and $\ti_{\gamma-2\epsilon}(\cP_{D_1})$ on the ${\cal L}$-part of $\rho$, and gets the collapsed state $\rho''$ such that $\Tr_{\cal R}[\rho'']=\Tr_{\cal R}[\rho_3]$.
    \item ${\cal R}$ measures $\ati^{\epsilon, \delta}_{\cP_2, D_2, \gamma}$ on $\Tr_{\cal L}[\rho'']$ and gets $\rho_3$.
    \item ${\cal R}$ measures $\ti_{\gamma-2\epsilon}(\cP_{D_2})$ on $\rho_4$.
\end{enumerate}

By \cref{cor:ati_thresimp} again, we have
\begin{align*}
    \Pr[{\cal R}\leftarrow 0|{\cal L}\leftarrow 0] = \Tr[\ti_{\gamma -2\epsilon}(\cP_{D_2})\cdot \rho_4] \geq 1-2\delta.
\end{align*}

Therefore, we have
\begin{align*}
     \Tr\big[ \big(\ti_{\gamma-2\epsilon}(\cP_{D_1})\otimes \ti_{\gamma-2\epsilon}(\cP_{D_2})\big) \rho'\big] \geq (1-2\delta)^2\geq 1-4\delta,
\end{align*}
which completes the proof of the second part of the lemma. \qed
\end{proof}

Notice that \cref{lem:ati_2d} can be easily generalized to the case of $q$-partite state. We state the following corollary without proof:

\begin{corollary} \label{cor:ati_qd}
Let $\cP_1,\cP_2,\dots,\cP_q$ be $q$ collections of projective measurements and $D_i$ be any probability distributions defined on the index set of $\cP_i$ for all $i\in [q]$. 
For any $0<\epsilon, \delta, \gamma < 1$, for all $i\in [q]$, the algorithms $\ati^{\epsilon, \delta}_{\cP_i, D_i, \gamma}$ satisfy the followings:
\begin{itemize}
    \item For any $q$-partite (possibly entangled, mixed) quantum state $\rho\in \mathscr{H}_{1}\otimes \cdots \otimes \mathscr{H}_{q}$,
    \begin{align*}
        \Tr\left[\left(\bigotimes_{i=1}^q \ati_{\cP_i, D_i, \gamma-\epsilon}^{\epsilon, \delta}\right)\rho\right] \geq \Tr\left[ \left(\bigotimes_{i=1}^q \ti_{\gamma}(\cP_{D_i})\right)\rho \right] - q\delta.
    \end{align*}
    \item For any (possibly entangled, mixed) quantum state $\rho$, 
        let $\rho'$ be the collapsed state after applying  $\bigotimes_{i=1}^q \ati_{\cP_i, D_i, \gamma-\epsilon}^{\epsilon, \delta}$ on $\rho$  (and normalized). Then, 
        \begin{align*}
            \Tr\left[ \left(\bigotimes_{i=1}^q \ti_{\gamma-2\epsilon}(\cP_{D_i})\right) \rho'\right]\geq  1- 2q \delta.
        \end{align*}
\end{itemize}
\end{corollary}

\subsection{Proof Sketch of \cref{thm:watermarktocd}}
\label{sec:cd_q_sketch}

We briefly sketch the proof for $q$-collusion resistance which is very similar to the case $q = 1$. Let $(\tilde{f}_1, \cdots, \tilde{f}_{q+1}, \sigma)$ be the output of the adversary. Let  $s_1, \cdots, s_q$ be the serial numbers in the \generate{} procedure.  Let $s'_1, \cdots,  s'_{q+1}$ be the serial numbers corresponding to $\sigma$.  If $\As$ succeeds, there are two cases: 
\begin{enumerate}
    \item $\left\{ s'_i \right\}_{i \in [q + 1]} \subseteq \left\{ s_j \right\}_{j \in [q]}$: in this case, $\As$ successfully copies one of the money state. Thus, we can use $\As$ to construct an adversary for the quantum money scheme. 
    
    \item $\left\{ s'_i \right\}_{i \in [q + 1]} \subsetneq \left\{ s_j \right\}_{j \in [q]}$: in this case, $\As$ successfully unmarks one of the marked program.   Thus, we can use $\As$ to construct an adversary for the watermarking scheme. 
\end{enumerate}

Therefore, assuming the existence of $q$-collusion resistant quantum money scheme and watermarking scheme, the construction above is a $q$-collusion resistant copy-detection scheme. 
\section{Public-key Quantum Money from Copy Detection}
In this section, we show that  we can use quantum copy detection and public-key encryption to construct a public-key quantum money scheme. This implication shows one more application of copy detection and further demonstrates the relationship between copy detection and public-key quantum money.

We  give the following the construction of the public-key quantum money.
Assume that we have an underlying public key encryption scheme called $\pke = (\pke.\keygen, \pke.\enc, \pke.\enc)$ with message space $\cM$, and an underlying copy detection scheme $\cd = (\cd.\setup, \cd.\generate,  \cd.\compute, \cd.\chk)$.

\begin{figure}[H]
    \centering
    \begin{gamespec}
    \begin{description}
    \item $\keygen(1^\lambda) \to (\pk,\sk):$ 
   
    \begin{itemize}
        \item Take in security parameter $\lambda$
        \item  Run $\pke.\keygen(1^\lambda) \to (\pke.\pk, \pke.\sk)$ and $\cd.\setup(1^\lambda) \to (\cd.\pk, \cd.\sk)$.
        \item Output $\pk = (\pke.\pk, \cd.\pk)$ and $\sk = (\pke.\sk, \cd.\sk)$.
    \end{itemize}

    \item $\gennote(\sk) \to \ket{\$}:$
    \begin{itemize}
        \item Take in the secret key $\sk = (\pke.\sk, \cd.\sk)$.
        
        \item Run $\cd.\generate(\cd.\sk, f = \pke.\dec(\pke.\sk, \cdot)$ to generate a copy detection program $(\rho_f, \{U_{f, x}\}_{x \in [N]})$ for the function $f = \pke.\dec(\pke.\sk, \cdot)$.
        
        \item Output $\ket{\$} = (\rho_f, \{U_{f, x}\}_{x \in [N]})$.
    \end{itemize}

    
    \item $\ver(\pk, \ket{\$'}) \to 0/1:$ \begin{itemize}
        \item Take in the public key $\pk = (\pke.\pk, \cd.\pk)$ and a claimed banknote state $\ket{\$'}$, i.e. a claimed copy detection program for $f = \pke.\dec(\pke.\sk, \cdot)$.
        
        \item Parse the claimed banknote  $\ket{\$'}$ as $(\aux_f, \rho_f’, \{U_{f, x}'\}_{x \in [N]})$.
        
        \item Run $\cd.\chk(\cd.\pk, \aux_f, \rho_f', \{U_{f, x}'\}_{x \in [N]})) \to b$; if $b = 1$, output 1 (for reject).
        
        \item Test if the program $(\rho_f', \{U_{f, x}'\}_{x \in [N]})$ is a $\gamma$-good program with respect to $f$, $E_\lambda$, using the public information in $\pk  = (\pke.\pk, \cd.\pk)$; if yes, output 0; else output 1.
    \end{itemize}

    \end{description}
    \end{gamespec}
    \caption{Public-key Quantum Money Scheme from Copy Detection}
    \label{fig:pk_money_I}
\end{figure}
\subsubsection{Security Analysis}
We now show that the public-key quantum money construction has correctness and unclonable security, given a quantum copy detection scheme with correctness and $\gamma$-anti-piracy security. The proof is intuitive and we omit some details.

\paragraph{Verification Correctness}

By the computation correctness of the underlying copy detection scheme $\cd$ and decryption correctness of the underlying $\pke$, a valid banknote $\ket{\$} = (\rho_f, \{U_{f, x}\}_{x \in [N]})$ for $f = \pke.\dec(\pke.\sk, \cdot)$ is supposed to pass $\chk$ and  be a $\gamma$-good program with respect to $f, E_\lambda$ with all but negligible probability.
Therefore, verification correctness holds. 


\paragraph{Unclonable Security} We  give a brief proof for the unclonable security of the quantum money scheme, whose security definition is given in \cref{def:money}. 

\begin{lemma}
\label{lem:money_pke_secure}
Assuming that the quantum copy-protection scheme $\cd$ has $\gamma$-anti-piracy, then public-key quantum money scheme has unclonable security.
\end{lemma}

\begin{proof}
Suppose there is a QPT adversary $\cA$ that breaks unclonable security, then we can construct a QPT adversary $B$ that breaks $\gamma$-anti-piracy security for $\cd$.

The quantum copy detection challenger interacts with $B$ in a copy detection anti-piracy game: In the Setup phase, challenger runs the setup $\cd.\setup(1^\lambda)$ to generate the keys $(\cd.\pk, \cd.\sk)$. In the Sampling phase, the challenger samples $f = \pke.\dec(\pke.\sk, \cdot)$, where
$(\pke.\pk, \pke.\sk) \gets  \pke.\keygen(1^\lambda)$; note that it gives $\aux = \pke.\pk$ to adversary $B$ and $s_f = \pke.\sk$ is kept secret.
$B$ then gives $(\cd.\pk, \pke.\pk)$ to the quantum money adversary $\cA$ as the public key.
In the Query phase, copy detection challenger generates one copy of copy detection program $(\rho_f, \{U_{f, x}\}_{x \in [N]}) \gets \cd.\generate(\cd.\sk, f)$ and gives to $B$. Then $B$ sends $(\rho_f, \{U_{f, x}\}_{x \in [N]})$ as a money state $\ket{\$}$ to $\cA$. Finally, $\cA$ output two claimed money states $\{\ket{\$_i}\} = \{\rho_i, U_{i}\}_{i \in [2]}$ and sends to $B$. $B$ uses them as its pirate programs and passes to copy detection challenger. 
It is easy to see that if both claimed money states $\{\ket{\$_i}\}_{i \in [2]}$ produced by $\cA$'s pass verification with non-negligible probability, then $B$ wins the copy detection anti-piracy security game with non-negligible probability.
\end{proof}

\end{document}